\documentclass[aip,jmp,amsmath,amssymb,reprint,onecolumn]{revtex4-1}
\usepackage{amsmath,amssymb,amsthm}
\usepackage{bbm}
\usepackage{graphicx}
\usepackage{placeins}
\usepackage{xcolor}
\usepackage{hyperref}
\usepackage{tikz}
\usetikzlibrary{calc,decorations,decorations.text,decorations.pathmorphing,shapes.callouts,shapes.symbols}
\hypersetup{colorlinks=true,citecolor=black,filecolor=black,linkcolor=black,urlcolor=black
    colorlinks=true, linktoc=all, linkcolor=blue,}
   
\newcommand{\EE}{\mathbb{E}}
\newcommand{\bbP}{\mathbb{P}}

\newcommand{\spacecase}[0]{\vspace{0.3cm} \\}
\newcommand{\Tr}{{\rm Tr}}

\newtheorem{theorem}{Theorem}
\newtheorem{prop}[theorem]{Proposition}
\newtheorem{definition}[theorem]{Defintition}
\newtheorem{lemma}[theorem]{Lemma}
\newtheorem{conjecture}[theorem]{Conjecture}
\newtheorem{hypothesis}[theorem]{Hypothesis}

\def\E{\mathbb{E}}

\def\cE{\mathcal E}
\newcommand{\tbf}[1]{{\bold{#1}}}

\begin{document}

\title{Storage capacity in symmetric binary perceptrons \vspace{0.5cm}}
\author{Benjamin Aubin}
\affiliation{Institut de physique th\'eorique, Universit\'e Paris
  Saclay, CNRS, CEA Saclay, F-91191 Gif-sur-Yvette, France
}
\author{Will Perkins}
\affiliation{Department of Mathematics, Statistics and Computer Science, University of Illinois, Chicago, USA}

\author{Lenka Zdeborov\'a}
\affiliation{Institut de physique th\'eorique, Universit\'e Paris
  Saclay, CNRS, CEA Saclay, F-91191 Gif-sur-Yvette, France
}

\date{\today}
\begin{abstract}
We study the problem of determining the capacity of the binary perceptron for two
  variants of the problem where the corresponding constraint is
  symmetric. We call these variants the rectangle-binary-perceptron (RPB)
  and the $u-$function-binary-perceptron (UBP). We show that, unlike for the
  usual step-function-binary-perceptron, the critical capacity in these
  symmetric cases is given by the annealed computation in a large region of parameter space (for all rectangular constraints and for narrow enough $u-$function constraints,
 $K<K^*$).  We prove this fact (under two natural assumptions) using the first and second moment
methods.  We further use the second moment method to conjecture that solutions of the symmetric binary perceptrons are organized in a so-called
frozen-1RSB structure, without using the replica method. We then use
the replica method to estimate the capacity threshold for the UBP case
when the $u-$function is wide $K>K^*$. We conclude that
full-step-replica-symmetry breaking would have to be evaluated in
order to obtain the exact capacity in this case.
\end{abstract}
\maketitle

\section{ Introduction }

In this paper we revisit the problem of computing 
the capacity of the binary perceptron \cite{1,2} for storing random
patterns. This problem lies at the core of early statistical physics studies
of neural networks and their learning and generalization properties,
for reviews see
e.g. \cite{watkin1993statistical,seung1992statistical,engel2001statistical,NishimoriBook01}.
While the perceptron problem is motivated by studies of simple
artificial neural
networks as discussed in detail in the above literature, in this
paper we view it as a random constraint satisfaction problem (CSP) where the
vector of binary  weights $\textbf{w} \in \{ \pm 1\}^N$ (a \textit{solution}) must satisfy $M$
{\it step} constraints of the type 
\begin{equation}
           \sum_{i=1}^N  X_{\mu i} w_i \ge K\, ,  \label{eq_constraints}
\end{equation}
where $\mu=1,\dots,M$, $K\in \mathbbm{R}$ is the \textit{threshold}, the random variables $X_{\mu i}$ are  $iid$ Gaussian
variables with zero mean and variance $1/N$, and the rows of the matrix
$\tbf{X}\in \mathbb{R}^{M \times N}$ are called patterns.  
We define an indicator function associated to the perceptron with
a step constraint as $\varphi^s(z) =
\mathbbm{1}_{\displaystyle z \geq K }$.

We say that a given vector $\textbf{w}$ is a solution of the
perceptron instance if all  $M$ constraints given by eq.~\eqref{eq_constraints}
are satisfied. The {\it storage capacity} is then defined similarly to
the satisfiability threshold in random constraint satisfaction
problems: we denote the constraint density as $\alpha \equiv M/N$
and define the storage capacity $\alpha_c(K)$ as the infimum
of densities $\alpha$ such that in the limit $N\to \infty$, 
 with high probability (over the choice of the matrix $\tbf{X}$) there are no solutions. It is natural to conjecture that the converse also holds, i.e. the storage capacity $\alpha_c(K)$ equals the supremum of $\alpha$  such  that in the limit $N\to \infty$ solutions exist with high probability.  In this case we would say the storage capacity is a \textit{sharp threshold}.

Gardner and Derrida in their paper \cite{1} assume the
storage capacity $\alpha_c(K)$ is a sharp threshold and they apply the replica calculation to
compute it, but reach a result inconsistent with a simple upper bound
obtained by the first moment method. M\'ezard and Krauth \cite{2} found
a way to obtain a consistent prediction from the replica calculation
and concluded that the storage capacity $\alpha^s_c(K)$ for the step binary perceptron
(SBP), i.e. associated to the constraint $\varphi^s$, is given by the largest $\alpha$ for which the following quantity, the  \textit{entropy} in physics, is positive:
\begin{equation}
\phi_{\rm RS}^{s}(\alpha,K) = \textrm{extr}_{q_0,\hat{q}_0}\left\{ \frac{1}{2}\left(q_0 - 1\right)\hat{q}_0 +  \int Dt \log \left[ 2 \cosh \left( t\sqrt{\hat{q}_0}\right)   \right ]
 + \alpha \int Dt \log\left[  \int^{\infty}_{\frac{K-t\sqrt{q_0}}{\sqrt{1-q_0}}} Du      \right] \right\} \, , 
 \label{RS_capacity}
\end{equation}
where $Dt = \frac{e^{-t^2/2}}{\sqrt{2\pi}} dt$ is a normal Gaussian
measure, and "$\rm extr$" means that the
expression is evaluated where the derivatives on the curl-bracket, with
respect to $q_0\ge 0$ and $\hat q_0 \ge 0$, are zero.

Several decades of subsequent research in the statistical physics of disordered systems are consistent with the conjectured M\'ezard-Krauth formula for the storage capacity of the binary perceptron. 
Despite the simplicity of the above conjecture and decades of impressive progress in the mathematics of spin glasses and related problems, (see e.g. \cite{talagrand2006parisi,talagrand2003spin,8,achlioptas2011solution,panchenko2014parisi,ding2015proof} and many others), the storage capacity of the binary perceptron remains an open mathematical problem. In fact, even the very existence of a sharp threshold, i.e. the fact that in the limit $N\to \infty$ the probability that patterns can be stored drops sharply from one to zero at the capacity, is an open problem. 
Up to very recently only widely non-matching upper bounds and lower
bounds for the storage capacity of the binary perceptron were available
\cite{kim1998covering,stojnic2013discrete}. As the present work was
being finalized Ding and Sun \cite{ding2018capacity} proved in a remarkable paper a lower
bound on the capacity that matches the
Krauth and Mezard conjecture (note that much like Theorem~\ref{thmMain} below, the main theorem in~\cite{ding2018capacity} depends on a numerical hypothesis).  A matching upper bound remains an open challenge in
mathematical physics and probability theory. 

In this paper we introduce two simple \textit{symmetric} variants of the binary
perceptron problem.  Let $z_\mu (\textbf{w}) = \sum_{i=1}^N  X_{\mu i} w_i$.  For a threshold $K\in \mathbbm{R}^+$, we consider
two different types of symmetric constraints: 
\begin{itemize}
	\item The rectangle binary perceptron (RBP) requires $|z_\mu|\le K, 
          \forall \mu=1,\dots,M$. Its associated indicator function is $\varphi^r(z) = \mathbbm{1}_{\displaystyle |z| \leq K }$.
	\item The $u$-function binary perceptron (UBP)
          requires $|z_\mu|\ge K, \forall \mu=1,\dots,M $. Its associated indicator function is $\varphi^u(z) = \mathbbm{1}_{\displaystyle |z| \geq K }$.
\end{itemize}
These constraints are symmetric in the sense that if $\textbf{w}$ is a
solution then $-\textbf{w}$ is a solution as well. Our motivation
behind these symmetric variants of the perceptron is that this
symmetry simplifies greatly the mathematical treatment of the problem,
while keeping the relevant physical properties intact. Thus, results
that remain open questions for the canonical perceptron can be established
rigorously for these symmetric versions. Symmetric perceptron models are also directly related to the  problem of determining the
discrepancy of a random matrix or set system 
\cite{BansalSpencer19}, a problem of interest in combinatorics. 

 
The main result of the present paper, presented in section~\ref{section:proof}, is a proof, subject to a numerical hypothesis, of a formula for the storage
capacity, defined in the same way as for the step-function binary
perceptron above. In particular, we show that in these
symmetric variants the first
moment upper bound (corresponding to the annealed capacity in physics)
on the storage capacity is tight (except for $K > K^* \simeq 0.817$
for the UBP case). We prove this statement using the
second moment method. 
We note that the existing physics literature on perceptron-like
problem contains other cases of models where the first moment upper
bound on the storage capacity was observed to be tight, in particular the parity
machine \cite{opper1995statistical}, and the reversed-wedge binary
perceptron \cite{bex1995storage,hosaka2002statistical}. Those works,
however, rely on the comparison of the first moment bound on the capacity with the
result of the replica method, rather than providing a rigorous justification. 

To formally state our main result,
let $Z \sim  \mathcal{N}(0,1)$, and for $K \in \mathbbm{R}^{+}$ let 
$p_{r,K}= \bbP[|Z| \le K]$ and $p_{u,K} = \bbP[|Z| \ge K]$. 
\begin{itemize}
	\item The storage capacity for the rectangle binary perceptron is:
	\begin{equation}
 \alpha_c^{r}(K) = \frac{-\log(2)}{\log(p_{r,K})} \hspace{0.5cm}
 \forall K \in \mathbb{R}^+  \, . \label{cap_rec}
\end{equation}
 \item The storage capacity for the $u-$function binary perceptron is: 
 \begin{equation}
      \alpha_c^{u}(K)= \frac{-\log(2)}{\log(p_{u,K})} \hspace{0.5cm}  {\rm for }\hspace{0.2cm} 0<K < K^* \simeq 0.817\, . \label{cap_ss}
\end{equation}
\end{itemize}
The constant $K^* \simeq 0.817 $ stems from the properties of the
second moment entropy eq.~\eqref{main:AT_second_moment}. In the physics terms it is defined as the point of intersection between the annealed
capacity $\alpha_a^{u}(K)$ and the local stability of the RS solution $\alpha_{\rm AT}^{u}(K)$ eq.~\eqref{main:AT_crossover_RS}. That is,  $K^*$ is the solution of the following equation:
\begin{equation}
        \pi p_{u,K}^2e^{K^2}\log(p_{u,K}) =-2 \log (2) K^2\, .
         \label{AT_crossover_RS}
\end{equation}

The two symmetric variants of the perceptron problem  considered here share many
of the intriguing geometric properties of the original step-function
binary perceptron problem. Most significant  is  the conjectured frozen-1RSB \cite{2} nature of the
space of solutions that splits into well separated clusters of vanishing entropy at
any $\alpha>0$. Remarkably, this frozen-1RSB property can  be deduced
from the form of the second moment entropy as we explain in section
\ref{section:frozen}. Our justification of the frozen-1RSB property does not rely on
the replica method and is hence of independent interest. 

For the UBP and $K > K^*$, the second-moment proof technique fails, and
this failure marks tightly the onset of the replica symmetry breaking
region. In that region, we evaluate the one-step replica symmetry
breaking (1RSB) approximation for the storage capacity, but conclude
that full-step replica symmetry breaking (FRSB) would be needed to
obtain the exact result. While the FRSB equations can be written along
the lines of \cite{20}, they are more involved than the ones for the
Sherrington-Kirkpatrick model \cite{parisi1979infinite,parisi1980sequence,parisi1980order}, and solving them numerically or
getting additional insight from them is a challenging task left for
future work. We present the replica analysis in section~\ref{section:replicas}. Table \ref{tab_summary} contains the summary of our main results along with the
predictions for the step-function perceptron. 

\setlength{\tabcolsep}{11pt}
\renewcommand{\arraystretch}{1.5}
\begin{table}
\centering
\begin{tabular}{|c||c|c|c|c|}
\hline
Binary perceptron & Constraint & Constraint function & Range of $K$ & Storage capacity \\
\hline
Step-function   & $z \geq K$      & $\varphi^s(z)=\mathbbm{1}_{\displaystyle z \geq K } $ \vspace{0.01cm}     & $\forall K \in {\mathbb R}$  & RS eq.~\eqref{RS_capacity}     \\
\hline
Rectangle   & $|z| \leq K$     & $\varphi^r(z) =\mathbbm{1}_{\displaystyle |z| \leq K } $ \vspace{0.01cm}      & $\forall K  \in {\mathbb R}^+$ & Annealed eq.~\eqref{cap_rec}       \\
\hline
$U$-function   & $|z|\geq K$   &  $\varphi^u(z)=\mathbbm{1}_{\displaystyle |z| \geq K } $ \vspace{0.01cm}     & $ 0< K<K^*=0.817$ & Annealed eq.~\eqref{cap_ss}
                                              \\
\hline
$U$-function &  $|z|\geq K$   &  $\varphi^u(z)=\mathbbm{1}_{\displaystyle |z| \geq K } $ \vspace{0.01cm}    & $\forall K>K^*=0.817$ & FRSB?           \\
\hline
\end{tabular}
\caption{This table summarizes results for storage capacity in binary
  perceptrons with different types of constraints. The result for canonical step-function is from \cite{2}. The results for the rectangle and $u$-function
  are obtained in this paper.}
\label{tab_summary}
\end{table}

Finally let us comment on the simpler and more commonly considered case of spherical perceptron where the binary constraint on the vector $\tbf{w}$ is replaced by the spherical constraint $\tbf{w}^\intercal\tbf{w} = \sum_{i=1}^N w^2_i = N$. For $K=0$ the spherical perceptron reduces to the famous problem of intersection of half-spaces with capacity $\alpha_c=2$ as solved by Wendell~\cite{wendel1962problem} and Cover \cite{cover1965geometrical}. For $K>0$ the Gardner-Derrida solution \cite{1}
 is correct as proven in \cite{shcherbina2003rigorous,stojnic2013another}. For $K<0$ the situation is more challenging and FRSB is needed to compute the storage capacity; for recent progress in physics see \cite{franz2016simplest,20}, while mathematical considerations about this case were presented in \cite{stojnic2013negative}.

\section{ Proof of correctness of the annealed capacity}
\label{section:proof}

To state the main results precisely we introduce some definitions.  Let $\tbf{X}({N,M})$ be the random $M \times N$ pattern matrix. Define the partition functions
\begin{align*}
	\mathcal Z_r(\tbf{X}) = \displaystyle \sum_{ \textbf{w} \in \{\pm 1\}^N} \prod_{\mu = 1}^M  \varphi^r ( \displaystyle z_\mu (\textbf{w}) ) 
	\hspace{0.3cm} \textrm{ and }\hspace{0.3cm}
	\mathcal Z_u(\tbf{X}) = \displaystyle \sum_{ \textbf{w} \in \{\pm 1\}^N} \prod_{\mu = 1}^M \varphi^u ( \displaystyle z_\mu (\textbf{w}) ) \,,
\end{align*}


which count respectively the number of solutions for the rectangle and $u-$function constraints respectively.
 Let  $\cE^r(N,M)$ and  $\cE^u(N,M)$ be the events that $\mathcal Z_r(\tbf{X})\ge1$  and $\mathcal Z_u(\tbf{X})\ge1$. 
 We  formally define the storage capacity.
\begin{definition}
The storage capacity $\alpha_c^r(K)$ is 
\begin{align*}
\alpha_c^r(K) &= \inf \{ \alpha: \lim_{N \to \infty} \bbP[ \cE^r(N,\lfloor \alpha N \rfloor ) ]   =0 \} \,,
\end{align*}
and likewise for $\alpha_c^u(K)$.
\end{definition}

It is believed that there is a sharp threshold for the existence of solutions.  
\begin{conjecture}
\label{conjSharp}
The storage capacity is a sharp threshold:
\begin{align*}
\alpha^r_c(K) &= \sup \{ \alpha: \lim_{N \to \infty} \bbP[ \cE^r(N,\lfloor \alpha N \rfloor ) ]   =1 \} \,,
\end{align*}
and likewise for $\alpha_c^u(K)$.
\end{conjecture}

The corresponding conjecture for the random k-SAT model is the celebrated `satisfiability threshold conjecture' proved for $k$ large by Ding, Sly, and Sun~\cite{ding2015proof}.

Next, couple two standard Gaussians $Z_1, Z_\beta$ by letting $Z$ and $Z'$ be independent standard Gaussians and setting $Z_1 = \sqrt{\beta} Z + \sqrt{1-\beta} Z'$ and $Z_\beta = \sqrt{\beta} Z - \sqrt{1-\beta} Z'$. 
Let 
\begin{align}
\begin{cases}
q_{r,K}(\beta) &= \bbP[ |Z_1 | \le K \wedge |Z_\beta| \le K ] = q_K
(\beta)\,,  \spacecase 
q_{u,K}(\beta) &=  \bbP[ |Z_1 | \ge K \wedge |Z_\beta| \ge K ] = 1 -
2p_{r,K} + q_K (\beta)  \label{qK}\,, 
\end{cases}
\end{align}

with $q_K(\beta)$ the probability that two standard Gaussians with correlation $2\beta-1$ are both at most $K$ in absolute value, that is:
\begin{align*}
		q_K (\beta) &=  \frac{1}{2\pi}
                              \int_{-K}^{K} dy
                              \int_{\frac{-K+(1-2\beta)y}{2\sqrt{\beta(1-\beta)}}}^{\frac{K+(1-2\beta)y}{2\sqrt{\beta(1-\beta)}}}
                              e^{-\frac{x^2+y^2}{2}} dx \, .      
\end{align*}

Note that $q_{t,K}(1) = p_{t,K}$ and $q_{t,K}(1/2) = p_{t,K}^2$ for $t \in \{r,u\}$.
We now  introduce the functions that dictate the effectiveness of the second moment bound.   Let
\begin{align}
F_{r,K, \alpha}(\beta) &=  H(\beta) +  \alpha \log q_{r,K}(\beta) \\ \vspace{0.3cm}
F_{u,K,\alpha}(\beta) &=  H(\beta) +  \alpha \log q_{u,K}(\beta)   \label{FK}
\end{align}
where $H(\beta) = -\beta \log \beta -(1-\beta) \log (1-\beta)$ is the Shannon entropy function.

We state a numerical hypothesis in terms of the derivatives of these two functions. 
\begin{hypothesis}
\label{hypo}
For all choices of $K>0$ and $\alpha>0$ so that $F''_{r,K,\alpha}(1/2) <0$, there is exactly one $\beta \in (1/2,1)$ so that $F'_{r,K,\alpha}(\beta) =0$.  The same holds for $F_{u,K,\alpha}$.
\end{hypothesis}

Our main theorem is a proof, under Hypothesis~\ref{hypo}, that the storage capacity is given by the annealed computation.
\begin{theorem}
\label{thmMain}
Under the assumption of Hypothesis~\ref{hypo}, the following hold. 
\begin{enumerate}
\item For all $K>0$, we have $\alpha_c^r(K) = -\log(2) / \log (p_{r,K})$. 
\item  For all $K \in (0, K^*)$, we have $\alpha_c^u(K) = -\log(2) / \log (p_{u,K})$.
\end{enumerate}
\end{theorem}

Under our definition of $\alpha_c^r(K)$ and $\alpha_c^u(K)$, we must prove two statements to show that $\alpha_c^r(K) = -\log(2) / \log (p_{r,K})$ (and similarly for $\alpha_c^u(K)$).  We use the first moment method to show that for $\alpha> -\log(2) / \log (p_{r,K})$, \\$\lim_{N \to \infty} \Pr(\cE^r(N,M)) =0$; then we use the second moment method to show that for $\alpha < -\log(2) / \log (p_{r,K})$, $\liminf_{N \to \infty} \Pr(\cE^r(N,M)) >0$ (a result analogous to what Ding and Sun prove for the more challenging step binary perceptron~\cite{ding2018capacity}).  Conjecture~\ref{conjSharp} asserts the stronger statement that for $\alpha < -\log(2) / \log (p_{r,K})$, $\lim_{N \to \infty} \Pr(\cE^r(N,M)) =1$.  


\subsection{First moment upper bound}
\label{proof:first_moment}

\begin{prop}
$ $ 
\begin{enumerate}
	\item If $\alpha > \alpha_a^{r}(K) =\frac{-\log(2)}{\log(p_{r,K}) } $, then whp there is no satisfying assignment to the binary perceptron with the rectangle activation function.
	\item If $\alpha > \alpha_a^{u}(K) =\frac{-\log(2)}{\log(p_{u,K}) } $, then whp there is no satisfying assignment to the binary perceptron with the $u$-function activation function.
	\end{enumerate}  
\end{prop}

\begin{proof}
We give the proof for the rectangle function as the proof for the $u$-function is identical. Let $\epsilon = \alpha - \alpha_a^{r}(K)  >0$. Let  $\tbf{1}$ denote the vector of dimension $N$ with all $1$ entries. 
\begin{align*}
\bbP[ \cE^r (N, \alpha N) ] &\le \E [\mathcal{Z}_{r}(\tbf{X}(N,\alpha N))] = 2^N \E \left [ \prod_{\mu =1}^{\alpha N} \mathbbm{1} _{\left |z_{\mu}(\tbf{1}) \right| \le K} \right] = 2^N p_{r,K}^{\alpha N} = \exp(N (\log(2)+\alpha \log( p_{r,K})  )) \\
&= \exp (N \epsilon \log (p_{r,K}))   \to 0  \text{ as } N \to \infty \,.
\end{align*}
\end{proof}

\subsection{Second moment lower bound}

\begin{prop} 
\label{prop2ndMoment}
$ $
\begin{enumerate}
\item If $\alpha < \frac{-\log(2)}{\log(p_{r,K}) }$, then  
$$ \liminf_{N\to \infty} \bbP[  \cE^r (N, \alpha N)  ] > 0 .$$ 
\item If $K < K^*$ and $\alpha <\frac{-\log(2)}{\log(p_{u,K}) }$, then 
$$ \liminf_{N \to \infty} \bbP[  \cE^u (N, \alpha N)   ]  >0.$$ 
\end{enumerate}
\end{prop}

To prove Proposition~\ref{prop2ndMoment}  we will apply the second-moment method in a similar fashion to Achlioptas and Moore~\cite{achlioptas2002asymptotic} who determined the satisfiability threshold of random $k$-SAT to within a factor $2$ by considering not-all-equal satisfying assignments (not-all-equal satisfiability (NAE-SAT) constraints are symmetric in the same way the rectangle and $u$-function constraints are symmetric). 
Recall the Paley-Zygmund inequality.
\begin{lemma}
Let $X$ be a non-negative random variable.  Then
\begin{align*}
\bbP[X >0 ] &\ge \frac{\E[X]^2 }{\E[X^2]} \,.
\end{align*}
\end{lemma}

We will also use the following application of Laplace's method from Achlioptas and Moore~\cite{achlioptas2002asymptotic}.
\begin{lemma}
\label{lemLaplace}
Let $g(\beta)$ be a real analytic function on $[0,1]$ and let 
\begin{align*}
G(\beta) &=  \frac{g(\beta)}{\beta^{\beta} (1-\beta)^{1-\beta} } \, .
\end{align*}
If $G(1/2) > G(\beta)$ for all $\beta \ne 1/2$ and $G''(1/2) <0$, then there exists constants $c_1, c_2$ so that for all sufficiently large $N$
\begin{align*}
c_1 G(1/2)^N \le \sum_{l = 0}^N  \binom{N}{l} g(l/N )^N  \le c_2 G(1/2)^N \,.
\end{align*}
\end{lemma}

\subsubsection{Rectangle binary perceptron}

We calculate 
\begin{align*}
\E [\mathcal Z_r (\tbf{X})^2   ] &=  \sum_{\tbf{w_1},\tbf{w_2}\in \{ \pm 1\}^N} \bbP[\tbf{w_1},\tbf{w_2}  \textrm{ satisfying}] = 2^N \sum_{\tbf{w} \in \{ \pm 1\}^N} \bbP[\tbf{1}, \tbf{w}  \textrm{ satisfying}]  = 2^N \sum_{l=0}^N \binom{N}{l} q_{r,K}(l/N)^{\alpha N} \, , 
\end{align*}
where we recall $q_{r,K}$ from eq.~\eqref{qK}.
Define 
\begin{align}
G_{r,K,\alpha}(\beta) &\equiv \exp( F_{r,K,\alpha}(\beta)) =  \frac{ q_{r,K}(\beta)^{\alpha} }{ \beta^{\beta} (1-\beta)^{1-\beta}  } \, ,
\label{eqGtoF}
\end{align}

If we can show that $G_{r,K,\alpha}(1/2) > G_{r,K,\alpha}(\beta)$ for all $\beta \ne 1/2$ and $G_{r,K,\alpha}''(1/2) <0$, then by Lemma~\ref{lemLaplace}, we have  
\begin{align*}
\E [\mathcal Z_r (\tbf{X})^2   ] & \le c_2 4^N q_{r,K}(1/2)^{\alpha N} \\
&=c_2 4^N  p_{r,K}^{2 \alpha N} \,.
\end{align*}


Then since $\mathcal Z_r (\tbf{X})$ is integer valued, we have 
\begin{align*}
\bbP[ \mathcal Z_r (\tbf{X}) \ge 1] &\ge \frac { \E [ \mathcal Z_r (\tbf{X})]^2 }{ \E [\mathcal Z_r (\tbf{X})^2   ]   }  =  \frac { ( 2^N  p_{r,K}^{\alpha N} )^2 }{ \E [\mathcal Z_r (\tbf{X})^2   ]   }   \\
&\ge  \frac{ ( 2^N  p_{r,K}^{\alpha N} )^2 }{c_2 4^N  p_{r,K}^{2 \alpha N}}  = 1/c_2 >0 \,.
\end{align*}

It remains to show that when $\alpha < \frac{-\log(2)}{\log(p_{r,K}) } $, then $G_{r,K,\alpha}(1/2) > G_{r,K,\alpha}(\beta)$ for all $\beta \ne 1/2$ and $G_{r,K,\alpha}''(1/2) <0$.  By eq.~\eqref{eqGtoF} and the fact that $G_{r,K,\alpha}'(1/2)=0$, it is enough to show the same for $F_{r,K,\alpha}$. 

Certainly one necessary condition is that  $F_{r,K,\alpha}(1/2) > F_{r,K,\alpha}(1)$.  This reduces to the condition $2 p_{r,K}^{2\alpha} > p_{r,K}^{\alpha}$ or $\alpha < \frac{-\log (2)}{\log( p_{r,K})}$  which is exactly the condition of Proposition~\ref{prop2ndMoment}. Next consider $F_{r,K,\alpha}''(1/2)$. 


A calculation shows that
\begin{align*}
F_{r,K,\alpha}''(1/2) &= 4 \left( -1 + \frac{2}{\pi} \frac{\alpha K^2 e^{-K^2} }{ p_{r,K}^2}  \right ) \,.
\end{align*}
In particular, $F_{r,K,\alpha}''(1/2) <0$ if and only if
\begin{align*}
\alpha &< \frac{\pi}{2} \frac{ p_{r,K}^2}{ K^2 e^{-K^2} } \,.
\end{align*}
But a calculation also shows that 
\begin{align*}
- \frac{\log(2) }{\log (p_{r,K})} <  \frac{\pi}{2} \frac{ p_{r,K}^2}{ K^2 e^{-K^2} }  
\end{align*}
for all $K>0$ and so the condition of Proposition~\ref{prop2ndMoment} implies that $F_{r,K,\alpha}''(1/2) <0$.  

Moreover, since $F_{r,K,\alpha}(\beta)$ is symmetric around $\beta=1/2$ and it has a local maximum at $\beta=1/2$, Hypothesis~\ref{hypo} implies that the global maximum of $F_{r,K,\alpha}(\beta)$ occurs at either $1/2$ or $1$, and since $F_{r,K,\alpha}(1/2) > F_{r,K,\alpha}(1)$, we have that  $F_{r,K,\alpha}(1/2) > F_{r,K,\alpha}(\beta)$ for all $\beta \ne 1/2$, completing the proof of Proposition~\ref{prop2ndMoment} for the rectangle binary perceptron.

\subsubsection{$u$-function binary perceptron}
\label{proof:second_moment_u}
The proof for the $u$-function is similar.  We can calculate
\begin{align*}
\E [\mathcal Z_u (\tbf{X})^2   ] &=   2^N \sum_{l=0}^N \binom{N}{l} q_{u,K}(l/N)^{\alpha N} = \exp\left( N (\log(2) +  F_{u,K,\alpha}(\beta) )  \right) \, ,
\end{align*}
where we recall $q_{u,K}$ from eq.~\eqref{qK}. 
Using Lemma~\ref{lemLaplace} and Hypothesis~\ref{hypo} again, it suffices to show that for $0 < K< K^*$ and $\alpha< \frac{-\log(2)}{\log ( p_{u,K})}$ we have $F_{u,K,\alpha}(1/2) > F_{u,K,\alpha}(1)$ and $F_{u,K,\alpha}''(1/2) <0$.  The first follows immediately from the fact that $\alpha< \frac{-\log(2)}{\log ( p_{u,K})}$.  For the second, we have
\begin{align*}
F_{u,K,\alpha}''(1/2) &= 4 \left (-1 +  \frac{2}{\pi} \frac{\alpha K^2 e^{-K^2} }{ p_{u,K}^2}   \right) 
\end{align*}
and so $F_{u,K,\alpha}''(1/2) < 0$ if and only if 
\begin{align*}
\alpha &< \frac{\pi}{2} \frac{ p_{u,K}^2}{ K^2 e^{-K^2} } \,. 
\end{align*}
Unlike with the rectangle function it is not true that 
\begin{align}
- \frac{\log(2)}{\log (p_{u,K})} < \frac{\pi}{2} \frac{ p_{u,K}^2}{ K^2 e^{-K^2} }	
\label{main:AT_second_moment}
\end{align}
 for all $K$: the left and right sides of the inequality cross at $K= K^*$, which  implicitly defines $K^*$.
 Thus for $K<K^*$ and $\alpha < - \frac{\log(2)}{\log (p_{u,K})}$ we have $F_{u,K,\alpha}''(1/2) < 0$, which completes the proof of Proposition~\ref{prop2ndMoment} for the $u$-function binary perceptron.

\begin{figure}[htb!]
\centering
\includegraphics[scale=0.32]{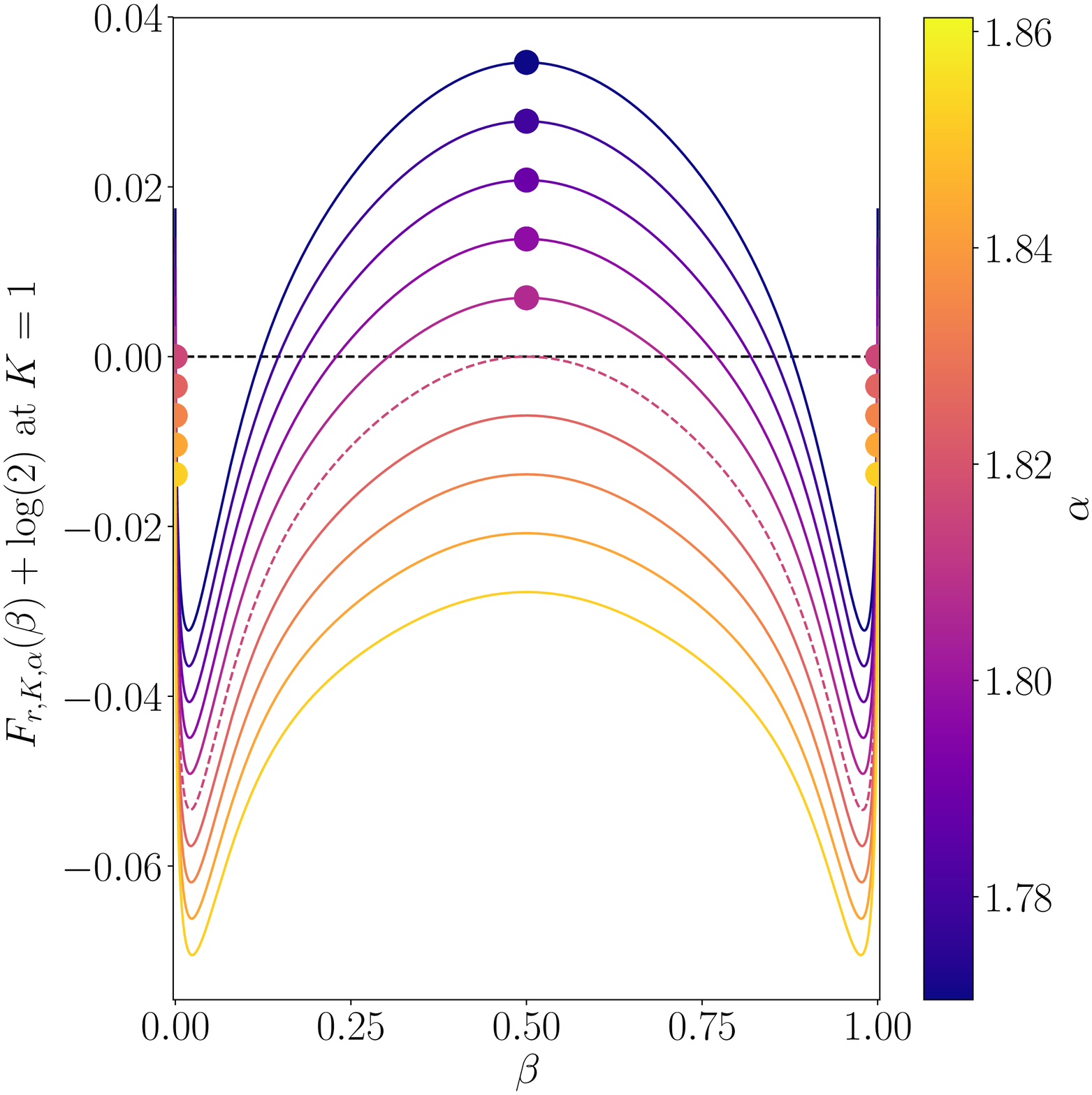}
\hspace{0.2cm}
\includegraphics[scale=0.32]{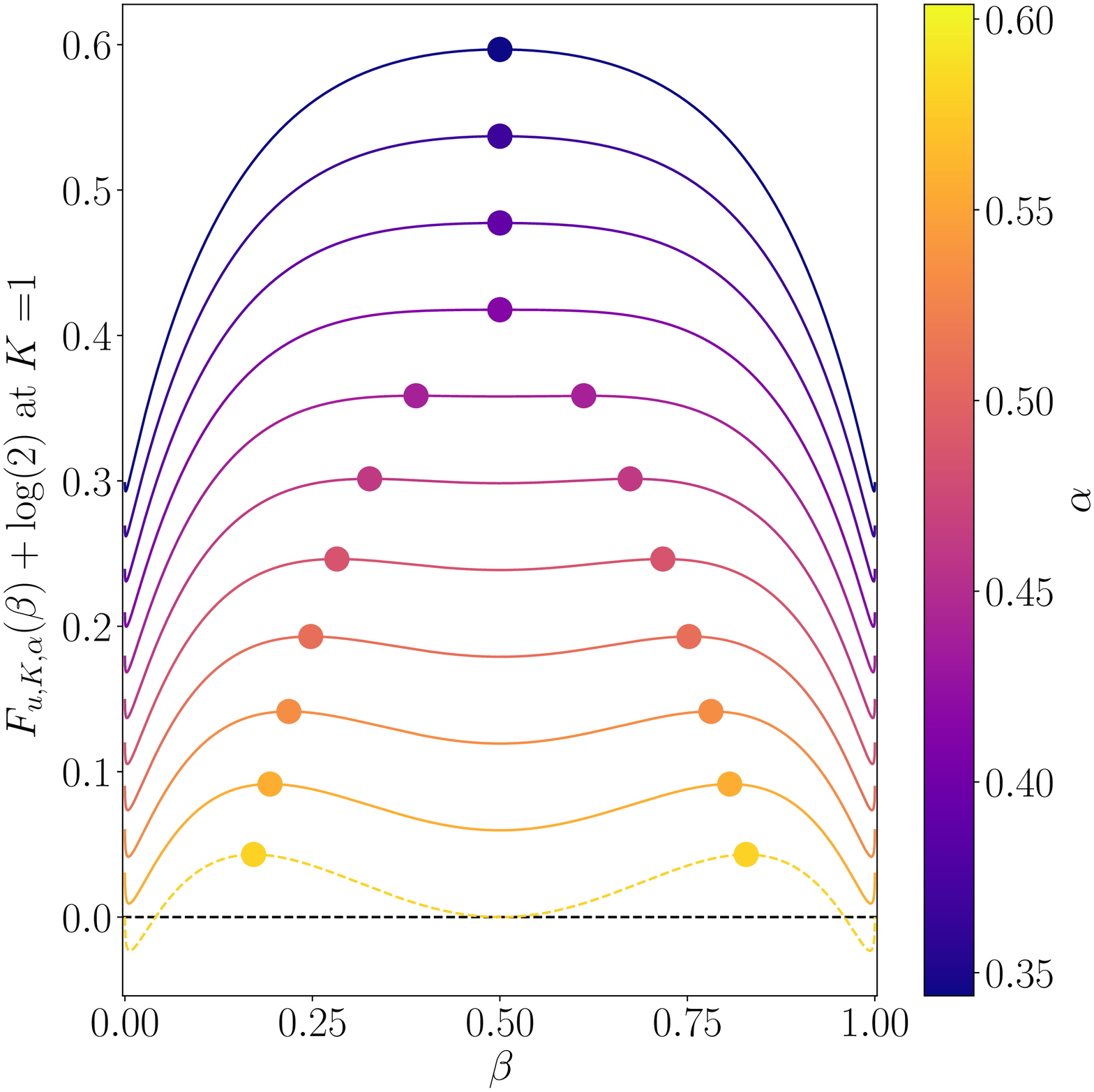}
\caption{Second moment entropy densities. \textbf{a)}: the rectangle binary perceptron for $\alpha \leq \alpha_a^r
  =1.816$ (dashed pink), $\beta=\frac{1}{2}$ is the global
  maximizer. For $\alpha \geq \alpha_a^r$, $\beta=0$ and $\beta=1$ are
  the maximizers. \textbf{b)}: the $u$-function binary perceptron for $\alpha \leq \alpha^*=0.430$, $\beta=\frac{1}{2}$ is the maximizer while for $\alpha^*\leq \alpha \leq \alpha_a^u = 0.604$ (dashed yellow), the maximizer is non-trivial $\beta \ne 0$.}
\label{main:plot_second_moment_rectangle_symstep}
\end{figure}
\FloatBarrier

\subsubsection{Illustration}
As an illustration, we plot the second moment entropy density $\lim_{N \to \infty} \frac{1}{N} \log \EE [ \mathcal{Z}_t^2 ] = \log(2) + F_{t,K,\alpha} $ for $t\in\{r,u\}$ at $K=1>K^*$ in fig.~\ref{main:plot_second_moment_rectangle_symstep}. For the rectangle function (\textbf{a}), the second moment is tight: the maximum is reached for $\beta=1/2$ for all $\alpha$ smaller than the first moment $\alpha_a^r$ (dashed pink). Exactly the same happens for the $u-$function with $K<K^*$. However for $K>K^*$, the second moment method fails (\textbf{b}): $\beta=1/2$ becomes a minimum and the maximum is obtained for non trivial values $\beta \ne 1/2$ for constraint density smaller than the first moment $\alpha_a^u$  (dashed yellow).

\section{Frozen-1RSB structure of solutions in binary perceptrons}
\label{section:frozen}

One of the most striking properties of the canonical step-function perceptron is
the predicted frozen-1RSB \cite{2} nature of the space of solutions. This means
that the dominant (measure tending to one) part of the space of solutions splits
into well separated clusters each of which has vanishing entropy
density at any $\alpha>0$.
This frozen-1RSB scenario and quantitative properties of the solution space were studied in detail recently \cite{16,huang2014origin}. Following up on conjectures
that such a frozen structure of solutions implies computational hardness in
diluted constraint satisfaction problems
\cite{zdeborova2008constraint}, it was argued that finding a satisfying assignment in the binary
perceptron should also be algorithmically hard since its solution space is dominated by
clusters of vanishing entropy density \cite{huang2014origin}. Yet this
conjecture contradicted empirical results of
\cite{braunstein2006learning}. This paradox was resolved in \cite{baldassi2015subdominant} where the authors identified that there
are subdominant parts (i.e. parts of measure converging to zero as
the system size diverges) of the solution space that form extended
clusters with large local entropy and all the algorithms that work well
always find a solution belonging to one of those  large-local-entropy
clusters. These sub-dominant clusters are not frozen and somewhat strangely are not captured in the canonical 1RSB calculation
\cite{baldassi2015subdominant}.   It was argued that existence of these
large-local-entropy clusters bears more general consequences on the
dynamics of learning algorithms in neural networks,
see e.g. \cite{baldassi2016unreasonable}. 

While frozen-1RSB structure has also been identified in constraint satisfaction
problems on sparse graphs \cite{zdeborova2008locked,zdeborova2011quiet}, we want to note
that its nature in the binary perceptron is of a rather different
nature. In sparse systems a simple argument using expansion
properties of the underlying graph and properties of the constraints show that each cluster with high
probability contains only one solution. In the perceptron model, which has a
fully connected bipartite interaction graph, this argument from sparse
models does not apply.

In the present paper, we deduce from the second moment
calculation of the previous section that the space of solutions in the symmetric binary
perceptrons is also of the frozen-1RSB type and this property moreover extends
to any finite temperature (with energy being defined as the number of
unsatisfied constraints). This is different from the locked constraint
satisfaction problems of \cite{zdeborova2008constraint,zdeborova2011quiet} living on
diluted hypergraphs, where the solution-clusters have extensive
entropy at any non-zero temperature. Another difference is that
whereas in the locked constraint satisfaction problems the size of
each cluster is one with high probability, in the binary perceptron
there are still many solutions in the clusters, it is only their entropy
density (i.e. logarithm of their number per variable) that vanishes
as $N\to \infty$. 

Investigation of the large local
entropy clusters and their implications for learning in the symmetric perceptrons is also
of great interest, but left for future work. Clearly since mathematically
the symmetric perceptrons are simpler than the step-function one, they
should also be the proper playground to deepen our understanding of
the large local entropy clusters and their relation to learning and generalization.

We present the frozen-1RSB scenario as a conjecture and then below indicate how the second moment calculation gives evidence for this conjecture.  Given an instance $\tbf{X}$ and a solution $\mathbf{w}$, let $\Gamma(\mathbf{w},d)$ denote the set of solutions $\mathbf{w}'$ with Hamming distance at most $d$ from $\mathbf{w}$.   

\begin{conjecture}
\label{conjFrozen}
For every $K > 0$ and every $\alpha \in (0,\alpha_c^r(K))$ there exists $d_{\text{min}}>0$ so that with high probability over the choice of the random instance $\tbf{X}$ from the RBP, the following property holds: for almost every solution $\mathbf{w}$, 
\begin{align*}
\frac{1}{N} \log | \Gamma(\mathbf{w},d_{\text{min}})| \to 0
\end{align*}
as $N \to \infty$.  The same holds for the UBP for all $K \leq K^*$.  
\end{conjecture}

\subsection{The link between the second-moment entropy and size of
  clusters} 

In this section we use $t \in\{r,u\}$ and note that the form of the second moment entropy density
$\frac{1}{N} \log \EE [ \mathcal{Z}_t^2 ]$ has very direct implications on the structure of
solutions in the corresponding models. 
As we defined it above, the second moment entropy is the normalized logarithm of
the expected number of
pairs of solutions of overlap $\beta$. 

For problems such as the symmetric binary perceptrons where the quenched
and annealed entropies are equal in leading order, there is a
striking relation between the planted and the random ensemble of the
model \cite{achlioptas2008algorithmic,krzakala2009hiding}. The random
ensemble is the problem we have considered so far, while the planted ensemble is
defined by starting with a configuration of the weights (a solution) and then
including only constraints that are satisfied by this {\it planted}
configuration. As long as the quenched and annealed entropies of the
random ensemble are equal in
leading order the planted and random ensembles should be  contiguous, meaning 
that high-probability properties that hold in one ensemble also hold in
the other. Moreover the planted configuration in the planted ensemble
has all the properties of a configuration sampled uniformly at random
in the random ensemble. These properties follow on the heuristic
level from the cavity method reasoning \cite{krzakala2009hiding}. They
were established fully rigorously in a range of models, see
e.g. \cite{achlioptas2008algorithmic,mossel2015reconstruction,coja2018information}. 
In the present case of symmetric binary perceptrons we have not yet managed
to prove contiguity between the random and the planted
ensemble, and so we leave a rigorous mathematical result for future work.  (In fact the missing ingredient is a version of Friedgut's sharp threshold result~\cite{friedgut1999sharp} suitable for perceptrons; such a result combined with Theorem~\ref{thmMain} would also prove Conjecture~\ref{conjSharp}).   We hence rely on the above heuristic argument and assume it
holds in what follows.

Given a planted solution $\mathbf{w}$ and a configuration $\mathbf{w}_\beta$  that agrees with $\mathbf{w}$  on $\beta N$ coordinates, the probability that $\mathbf{w}_\beta$ is a solution in the planted model is $( q_{t,K}(\beta)/ p_{t,K})^M$, and thus the expected number of solutions at Hamming distance $\beta N$ from the planted solution in the planted ensemble is 
\begin{align*}
 \EE[\mathcal Z_\beta ]= \binom{N}{\beta N} ( q_{t,K}(\beta)/ p_{t,K})^M \,,
\end{align*}
and its entropy density is 
\begin{align}
	\omega_t(\beta) \equiv \lim_{N \to \infty} \frac{1}{N} \log \EE[\mathcal{Z}_\beta]  = F_{t,K,\alpha}(\beta)- \alpha \log{p_{t,K}} \textrm{ for } t \in\{r,u\} \,.
\label{omega}
\end{align}


Recalling that contiguity  implies that the planted
solution has the properties of a uniformly chosen solution in the random ensemble then this entropy gives
us direct access to properties of the solution space in the random ensemble at equilibrium. Most
notably we notice (see derivation in section~\ref{frozen_2nd} below) that the derivative of
$\omega_t(\beta)$ at $\beta=1$ is $+\infty$ thus implying that $\forall
\epsilon>0$ with high probability there are no solutions at overlap $\beta \in [d_{\rm min}(\alpha,K),
(1-\epsilon)]$. In turn, this means that the dominant (measure
converging to one as $N\to \infty$) part
of the solution space splits into clusters each of which has
vanishing entropy density (i.e. logarithm of the number of solutions
in the cluster divided by $N$ goes to zero as $N\to \infty$).   The missing ingredient in a full proof of Conjecture~\ref{conjFrozen} is a proof of the contiguity statement.

\subsection{Form of the 2nd moment entropy implying frozen-1RSB}
\label{frozen_2nd}

\begin{figure}[htb!]
		    \centering
	   		\includegraphics[scale=0.32]{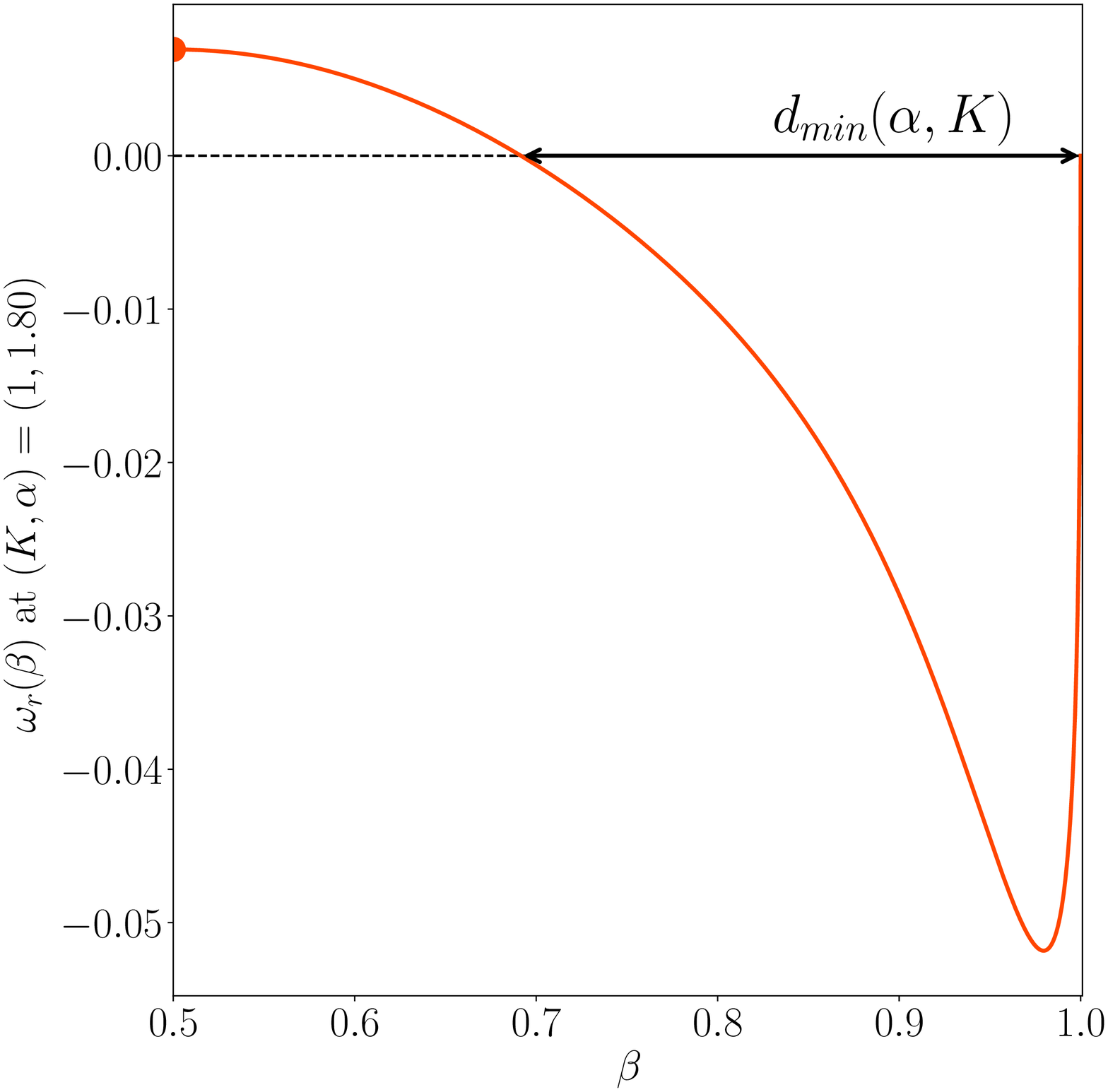}
	   		\hspace{0.2cm}
	   		\includegraphics[scale=0.32]{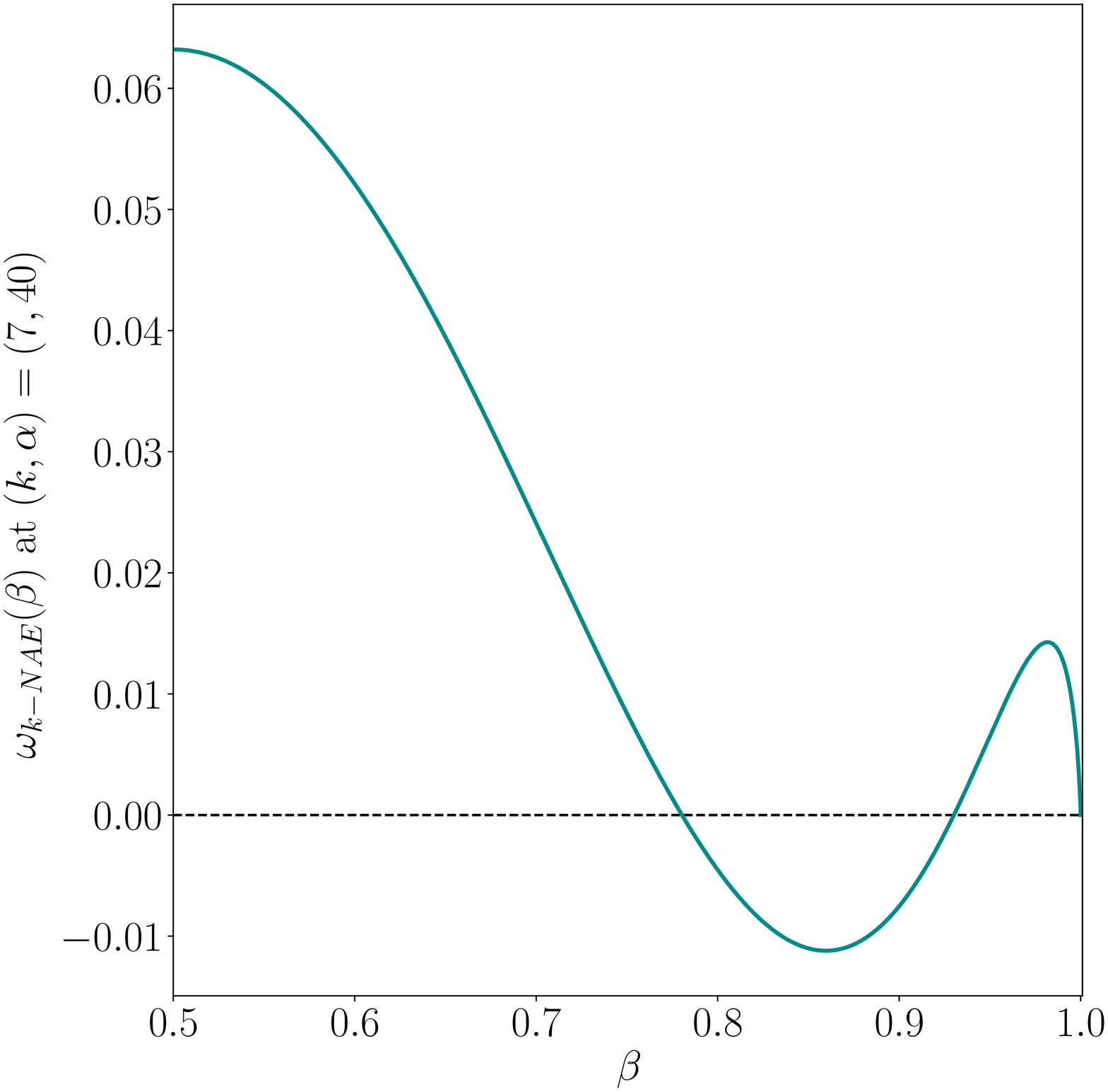}
\caption{\textbf{a)} Density of the annealed entropy of solutions at overlap $\beta$ from a
  random solution in the rectangle binary
  perceptron at $K=1$, $\alpha = 1.80 \le \alpha_c^{r}(K=1)$. We see there are no solution in an interval of overlaps $(1-d_{\rm min}, 1-\epsilon)$. This curve is
  obtained from the second moment entropy and contiguity between
  the random and planted ensembles. It implies the frozen-1RSB nature of the space of solutions. The same holds for the $u-$function. \textbf{b)} To compare we plot the density of the annealed entropy of solutions at overlap $\beta$ from a random solution in the $kˆ'$-NAE SAT model~\cite{achlioptas2002asymptotic} at $k = 7$, $\alpha = 40$. We see the density is positive in a large region close to $\beta = 1$, showing the absence of frozen-1RSB structure in this problem. 
  } 
 \label{fig_frozen}
\end{figure}

In fig.~\ref{fig_frozen}a we plot $\omega_r(\beta)$ for the rectangle
binary perceptron, at $K=1$,
$\alpha=1.80\le \alpha_c^{r}(K=1)$. Thanks to the contiguity between the planted and random ensembles that holds as long as the second moment entropy
density is twice the first moment entropy density, this curve
represents also the annealed entropy of solutions at overlap $\beta$
with a random reference solution. We see notably that there is an
interval of distances in which no solutions are present. Analytically
we can see from the properties of the functions $F_{t,K,\alpha} (\beta)$ and
$\log{p_{t,K}}$ that $F_{t,K,\alpha}(1) = \alpha \log{p_{t,K}}$ and the derivative of
$F_{t,K,\alpha}(\beta) \to \infty$. This is in contrast with, for instance, the
satisfiability problems studied in \cite{achlioptas2002asymptotic},
where the function corresponding to $F_{t,K,\alpha}(\beta)$ would have a negative
derivative in $\beta=1$ (see fig.~\ref{fig_frozen}b). There could still be an interval of {\it
  forbidden} distance, but the bump in entropy for $\beta \approx 1$
corresponds to the size of the clusters to which typical solutions
belong and those would be extensive.


\subsubsection{Frozen 1RSB in rectangle binary perceptron}

In the rectangle binary perceptron, the random and planted ensembles are conjectured to be contiguous for
all $K >0$ and $\alpha \in (0, \alpha^r_c(K))$. Using eq.~\eqref{FK}, the first derivative of $\omega_r(\beta)$, eq.~\eqref{omega}, is given by (see Appendix \ref{moments_finiteT})
\begin{equation*}
	 \frac{\partial \omega_r }{\partial \beta} = \frac{\partial F_{r,K,\alpha}}{\partial \beta} =  \log\left (\frac{1-\beta}{\beta} \right) +  \frac{\alpha}{ q_{r,K,T}(\beta)} \frac{1}{\pi \sqrt{\beta(1-\beta)}} \left( e^{-\frac{K^2}{2(1-\beta)}} \left( e^{\frac{(2\beta-1)K^2}{2(1-\beta)\beta}} -1  \right)  \right)\,,
\end{equation*}
and it diverges for all $K \in \mathbb{R}^+$, $\alpha>0$ in the limit $\beta \to 1$: 
\begin{equation}
	\frac{\partial \omega_r }{\partial \beta} \xrightarrow[\beta \to 1]{} +\infty \, . 
\end{equation}
This implies vanishing entropy density of clusters to which typical
solutions belong.


\subsubsection{Frozen 1RSB in the $u$-function binary perceptron}

In the $u$-function binary perceptron, the random and planted ensembles are conjectured to be contiguous for
all $0 < K \le K^*$ and $\alpha \in (0, \alpha^u_c(K))$. Using eq.~\eqref{FK}, the first derivative of $\omega_u(\beta)$ eq.~\eqref{omega}, is given by
\begin{align*}
	 \frac{\partial \omega_u }{\partial \beta} &= \frac{\partial F_{u,K,\alpha}}{\partial \beta} =  \log\left (\frac{1-\beta}{\beta} \right) +  \frac{\alpha}{ q_{u,K,T}(\beta)} \frac{1}{\pi \sqrt{\beta(1-\beta)}} \left( e^{-\frac{K^2}{2(1-\beta)}} \left( e^{\frac{(2\beta-1)K^2}{2(1-\beta)\beta}} -1  \right)  \right)\\
	 & \underset{\beta \to 1}{\longrightarrow} +\infty \, ,
\end{align*}
thus reaching the same conclusion on presence of frozen-1RSB.

In appendix \ref{moments_finiteT} we extend the second moment calculation to finite
temperature (for both the rectangle and $u-$function case). This means that we define energy of a configuration
${\cal E}(\tbf{w})$ as
the number of constraints that are violated by this
configurations. Then the corresponding partition function is defined
${\cal Z}(T) = \sum_{\tbf{w}} e^{-{\cal E}(\tbf{w})/T}$. There is a
one-to-one mapping between the temperature $T$ and energy density
$e={\cal E}/N$, consequently 
the corresponding finite-temperature second moment entropy density
counts the number of pairs of solutions at overlap $\beta$ and energy
density $e$. In appendix \ref{moments_finiteT} we apply the same
argument as here
connecting the random and planted ensemble, and deduce that  the
finite-temperature solution space of the models is of also of the frozen-1RSB type
for any $T<\infty$. 

\subsection{Frozen-1RSB as derived from the replica analysis}

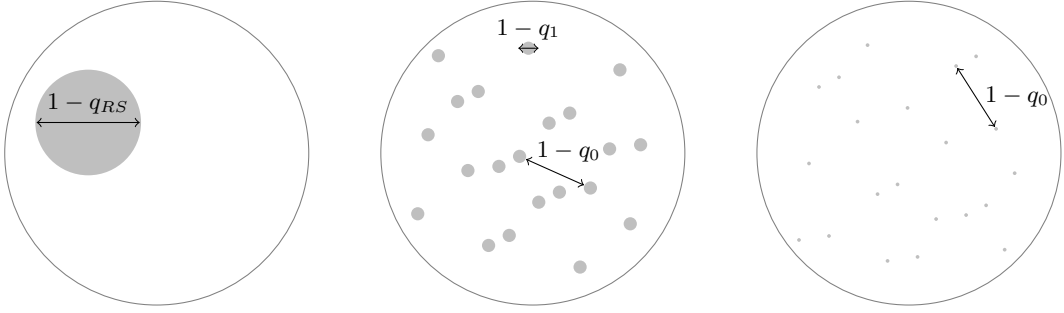
\begin{figure}[htb!]
	\begin{tikzpicture}
    \tikzstyle{factor}=[rectangle,fill=black,minimum size=7pt,inner sep=1pt]
    \tikzstyle{annot} = [text width=2.5cm, text centered]
    \tikzstyle{RSCluster} = [circle,fill=gray,minimum size=40pt,inner sep=1pt,fill opacity=0.5]]
    \tikzstyle{RSBCluster} = [circle,fill=gray,minimum size=5pt,inner sep=1pt,fill opacity=0.5]]
    \tikzstyle{fRSBCluster} = [circle,fill=gray,minimum size=0.5pt,inner sep=0.5pt,fill opacity=0.5]
    \tikzstyle{configSpace} = [circle,minimum size=115pt,inner sep=1pt,draw=gray, fill opacity=1.,fill=white]
    
    \node[configSpace] (C-1) at (-5,0) {};
    \node[configSpace] (C-1) at (0,0) {};
    \node[configSpace] (C-1) at (5,0) {};
    
    \pgfmathsetseed{3}
    \coordinate (mycenterpoint) at ($ (-5,0)+{rand}*(1,0)+ {rand}*(0,1)$);
    \node[RSCluster] (RS-1) at (mycenterpoint) {};
    \node (C) at ([shift=({-0.8cm,0cm})]RS-1) {};
	\node (D) at ([shift=({0.8cm,0cm})]RS-1) {};
	\draw [<->,black] (C) -- (D);
	\path [<->] (C) -- node [above] {$1-q_{RS}$} (D) ;
    
    \foreach \l in {10,...,30}
      {	\pgfmathsetseed{\l}
          \coordinate (mycenterpoint) at (${rand}*(1.55,0)+ {rand}*(0,1.55)$);
          \node[RSBCluster] (RSB-\l) at (mycenterpoint) {};
      }
    
    \node (A) at ([shift=({-0.25cm,0cm})]RSB-10) {};
	\node (B) at ([shift=({0.25cm,0cm})]RSB-10) {};
	\draw [<->,black] (A) -- (B);
    \path [<->] (A) -- node [above] {$1-q_{1}$} (B) ;

	\draw [<->,black] (RSB-28) -- (RSB-30);
    \path [<->] (RSB-28) -- node [above] { $1-q_{0}$} ++(-0.5,0.4) (RSB-30) ;
    
    \foreach \l in {40,...,60}
      {	\pgfmathsetseed{\l}
          \coordinate (mycenterpoint) at ($(5,0)+ {rand}*(1.5,0)+ {rand}*(0,1.5)$);
          \node[fRSBCluster] (fRSB-\l) at (mycenterpoint) {};
      }

    \draw [<->,black] (fRSB-48) -- (fRSB-60);
    \path [<->] (fRSB-48) -- node [above] { $1-q_{0}$} ++(1.6,-1.3) (fRSB-60) ;

    \end{tikzpicture}
	\caption{Illustration of the configuration space for the different phases: \textbf{a)}: RS - solutions are concentrated in a single cluster of typical size $1-q_{\rm RS}$. \textbf{b)}: 1RSB -  solutions form clusters of size $1-q_1$ at a distance $1-q_0$ from each other. \textbf{c)}: f1RSB - clusters are point-like ($1-q_1\simeq 0$) at a distance $1-q_0 = 1-q_{RS}$ from each other.}
	\label{configurationSpace}
\end{figure}

We stress that we derived the frozen-1RSB nature of the space of
solutions without the use of replicas. For completeness we summarize
here how this translates to the properties of the
one-step-replica-symmetry breaking solution. This is the way this
phenomena was originally discovered and described in
\cite{2,martin2004frozen,16}. For readers not familiar with the
replica method this section should be read after reading section~\ref{section:replicas}.

In general, three kinds of fixed points of the 1RSB equations are possible: 
\begin{itemize}
	\item The replica symmetric (RS) solution $q_0= q_1 = q_{\rm RS}<1$\,,
	\item The frozen-1RSB solution (f1RSB) $(q_0,q_1)=(q_{\rm RS},1)$\,,
	\item The 1RSB solution  $(q_0,q_1)$ with $q_1\ne 1$\,.
\end{itemize}

The frozen-1RSB is characterized by
an inner-cluster overlap $q_1=1$ and an inter-cluster overlap
$q_0=q_{\rm RS}$, which means that clusters have vanishing entropy density and
remain far from each other.
Mathematically RS and f1RSB solutions are equivalent in the sense that
these solutions have the same free energy eq.~\eqref{main:phi_1RSB} $\Phi_{\rm 1RSB}\{q_0 = q_{\rm RS}, q_1 = q_{\rm RS}\}=\Phi_{\rm 1RSB}\{q_0 = q_{\rm RS},q_1 = 1\}$, and the complexity
of the f1RSB solution equals the RS entropy $\Sigma(\phi=0) =
\phi_{\rm RS}$ eq.~(\ref{main:complexity}, \ref{main:phi_RS}). However, RS and f1RSB do
not share the same configuration space. The RS
phase is associated to a single cluster of solution with typical size
$1-q_{\rm RS}$, while the f1RSB configuration space is composed of many
point-like solutions of size $q_1\simeq 1$ and at distance $1-q_0 =
1-q_{\rm RS}$ of each other, see fig.~\ref{configurationSpace}. From this point of view f1RSB is the correct
description of the phase space.

\section{Replica calculation of the storage capacity}
\label{section:replicas}

In this section we recall the replica calculation leading to the expression
of the storage capacity in the step-function binary perceptron. We show that
in the symmetric binary perceptrons the annealed calculation is
reproduced by the replica symmetric result. For the $u-$function binary
perceptron we show that $K^*$ coincides with the onset of replica
symmetry breaking and we evaluate the 1RSB capacity for $K>K^*$.

\subsection{Replica calculation}

For the purpose of the calculations, we introduce the constraint function $\mathcal{C}(\textbf{z})$ that returns $1$ if  $\tbf{w}$ satisfies all the constraints \{$\varphi(z_\mu)\}_{\mu=1}^M$ and $0$ otherwise
\begin{align*}
\mathcal{C}(\tbf{z}) = \displaystyle \prod_{\mu=1}^M
        \varphi(z_{\mu})  \vspace{0.2cm} \textrm{ with }\vspace{0.2cm} z_{\mu}= \tbf{X}_{\mu}
        \tbf{w} \,.	
\end{align*}
Recall the partition function $\mathcal{Z}$ is the number of satisfying vectors $\tbf{w}$, with prior distribution $P_w(\tbf{w})$, for a given matrix $\tbf{X}$
\begin{equation*}
	\mathcal{Z}(\tbf{X}) = \displaystyle \sum_{\tbf{w} \in \{\pm 1\}^N} \prod_{\mu=1}^M  \varphi\left(\tbf{X}_{\mu}
        \tbf{w} \right)=
        \int d\tbf{w} P_w(\tbf{w})\int d\tbf{z} \,
        \mathcal{C}(\tbf{z})\delta(\tbf{z}-\tbf{X}\tbf{w}) \,.
\end{equation*}

The replica method allows one to compute explicitly the quenched average
$\mathbb{E}_{\tbf{X}}[\log(\mathcal{Z(\tbf{X})})]$ \cite{3}. More precisely, using the replica trick, the average of the logarithm can be expressed as the limit $n\to0$ of the derivative with respect to $n$ of the average of the $n$-th moment of the partition function. Finally the free entropy reads:
\begin{equation}
	\phi(\alpha) \equiv  \lim_{N\rightarrow +\infty} \frac{1}{N}
        \mathbb{E}_{\tbf{X}}[\log(\mathcal{Z(\tbf{X})})] =
        \lim_{N\rightarrow +\infty} \lim_{n\rightarrow 0} \frac{1}{N
          n}\frac{ \partial \log
          \left(\mathbb{E}_{\tbf{X}}[\mathcal{Z}(\tbf{X})^n]\right)}{\partial
          n} \, .
\label{main:free_entropy_trick}
\end{equation}

Computing the $n$-th moment of the partition function $\mathcal{Z}$,
for $n\in \mathbbm{N}$, is equivalent to considering $n$ copies, also
called replicas, of the initial system. For a given disorder, these $n$ replicas are non-interacting and $\mathcal{Z}^n$ can be computed easily. However, averaging over the "disorder" with distribution $P_X$ makes the replicas interacting: replicated weight-vectors $\tbf{w}^a$ and $\tbf{w}^b$, for $a,b \in [1:n]$, are correlated by the overlap matrix $\tbf{Q}=\left(Q_{ab}\right)_{a,b=1}^n=\left(\frac{1}{N} \sum_{i=1}^N w_i^a w_i^b\right)_{a,b=1}^n$.
 
We start averaging over the distribution $P_X$ and then use an analytical continuation for $n \in \mathbbm{R}$ and reverse the limits $N \to \infty$ and $n \to 0$. The exchange of limits $n \to 0$ and $N \to \infty$ is a key and classical ingredient for replica calculations, rendering the replica method heuristic and not rigorously justified. Using this later point, we show in Appendix \ref{appendix:general_replica_calculation} that the free entropy $\phi$ eq.~\eqref{main:free_entropy_trick} can finally be expressed as a saddle point equation over $n \times n $ symmetric matrices $\tbf{Q}$ and $\tbf{\hat{Q}}$
\begin{equation}
	\phi (\alpha) = -\text{SP}_{ \tbf{Q}, \tbf{\hat{Q}} }
        \left\{\lim_{n\rightarrow 0} \frac{\partial
            S_n(\bold{Q},\bold{\hat{Q}})}{\partial  n} \right\}\, ,
	\label{main:replica_general}
\end{equation}
where $\tbf{\hat{Q}}$ is a parameter involved in the change of variable between $\{\tbf{w}^a,\tbf{w}^b\}$ and $Q_{ab}$ and with
\begin{equation*}
     \begin{cases}
     S_n (\tbf{Q},\tbf{\hat{Q}} ) = \frac{1}{2}\Tr(\tbf{Q\hat{Q}})
     -\log(\mathcal{I}_w^n (\tbf{\hat{Q}}))-\alpha\log
     \left(\mathcal{I}_z^n(\tbf{Q}) \right)\, ,
      \vspace{0.2cm} \\
        \mathcal{I}_w^n (\tbf{\hat{Q}}) =  \int_{\mathbbm{R}^n}
        d\tbf{\tilde{w}} P_{\tilde{w}}(\tbf{\tilde{w}})  e^{
          \frac{1}{2}\tbf{\tilde{w}}^{\intercal} \tbf{\hat{Q}}
          \tbf{\tilde{w}} }  \hspace{0.3cm}  \textrm{ where }
        \tbf{\tilde{w}} \in \mathbbm{R}^n \textrm{ and } \displaystyle
        P_{\tilde{w}}(\tbf{\tilde{w}}) = \prod_{a=1}^n
        [\delta(\tilde{w}_a-1) + \delta(\tilde{w}_a+1)]  \,
        ,\vspace{0.2cm} \\
        \mathcal{I}_z^n(\tbf{Q}) =   \int_{\mathbbm{R}^n}
        d\tbf{\tilde{z}} P_{\tilde{z}}(\tbf{\tilde{z}})
        \mathcal{C}(\tbf{\tilde{z}}) \hspace{1.25cm}  \textrm{ where }
        \tbf{\tilde{z}} \in \mathbbm{R}^n \textrm{ and } P_{\tilde{z}} \triangleq
        \mathcal{N}\left( \tbf{0}, \tbf{Q} \right) \, .
     \end{cases}
\end{equation*}
In order to be able to compute the derivative of $S_n$ with respect to $n$ eq.~\eqref{main:replica_general}, we need an analytical formulation of $\tbf{Q}$ and $\tbf{\hat{Q}}$ as a function of $n$. 

\subsection{RS entropy}
The simplest ansatz is to assume that the overlap matrix $\tbf{Q}$ is Replica Symmetric (RS), which means that all replicas play the same role: the correlation between two arbitrary, but different, replicas is denoted $q_0$, and therefore the RS ansatz reads: 
\begin{equation*}
\forall (a,b) \in [1:n]\times [1:n], \hspace{0.1cm} \frac{1}{N} (\tbf{w}^a\cdot\tbf{w}^b) = 
	\begin{cases}
		q_0  \textrm{ if } a\ne b\, ,\\
		Q = 1 \textrm{ if } a = b \, .
	\end{cases}	
\end{equation*}
It enforces the matrix $\tbf{\hat{Q}}$ to present the same symmetry, respectively with parameters $\hat{q}_0$ and $\hat{Q}=1$. Using this ansatz and the $n\to 0$ limit, the Replica Symmetric (RS) entropy can be expressed as a set of saddle point equations over scalar parameters $q_0$ and $\hat{q}_0$, evaluated at the saddle point (Appendix \ref{appendix:replica_rs}): 
\begin{equation}
	\phi_{\rm RS}(\alpha)=  \textrm{extr}_{q_0,\hat{q}_0} \left\{
          -\frac{1}{2} + \frac{1}{2}(q_0\hat{q}_0-1) +  \mathcal{I}^w_{\rm
            RS}(\hat{q}_0)   +\alpha \mathcal{I}^{z}_{\rm RS}(q_0)
        \right\}\, ,
        \label{main:phi_RS}
\end{equation}

\begin{equation}
\textrm{ with }
	\begin{cases}
		\mathcal{I}^w_{\rm RS}(\hat{q}_0) \equiv  \int Dt\log\left( g_0^w (t,\hat{q}_0)\right) \, ,\vspace{0.5cm} \\
		\mathcal{I}^{z}_{\rm RS}(q_0) \equiv   \int Dt  \log\left( f_0^z (t,q_0)  \right)\, ,
	\end{cases}
	\textrm{ and for $i\in \mathbbm{N}$ }
		\begin{cases}
		g_i^w (t,\hat{q}_0) \equiv \displaystyle \int dw \, w^i P_w(w) \exp\left(
                  \frac{(1- \hat{q}_0  )}{2} w^2 + t\sqrt{\hat{q}_0} w
                \right)\, ,
			 \vspace{0.3cm} \\
		f_i^z(t,q_0) \equiv \displaystyle \int Dz \,  z^i \varphi(\sqrt{q_0} t +
                \sqrt{1-q_0} z)\, .
 	\end{cases}
 	\label{main:f_z_g_w_rs}
\end{equation}

Note that above and in what follows $Dt = \frac{e^{-t^2/2}}{\sqrt{2\pi}} dt$. In
  the binary perceptron case, the function $P_w$ is defined as
  $P_w(w)= [ \delta(w-1) + \delta(w+1) ]$ (note that this is not a
  probability distribution because of the normalization), and recall
  $\varphi(z)$ is the indicator function, checking that a constraint
  on the argument is satisfied (e.g in the step case, $\varphi^s(z) = 1$ if $z>K$).

While in the step binary perceptron (SBP) the fixed point solution
$(q_0,\hat{q}_0)$ is non-trivial, the symmetry of the activation
function in the RBP and UBP cases enforces the configuration space to
be symmetric and the fixed point $(q_0,\hat{q}_0)= (0,0)$ to exist. If
this symmetric fixed point is stable and has the lowest free energy,
the RS free entropy matches the annealed entropy
$\phi_a^t (\alpha) = \log(2)+\alpha \log(p_{t,K}) = \frac{1}{N} \log \EE_{\tbf{X}} [ \mathcal{Z}_t ( \tbf{X} ) ]$ from section~\ref{proof:first_moment} with $t \in \{r,u\}$. 

\subsubsection{Rectangle }
Solving numerically the corresponding saddle point equations leads to the single symmetric fixed point $(q_0,\hat{q}_0)=~(0,0)$. Hence the RS entropy saturates the first moment bound: 
\begin{align*}
	\phi_{\rm RS}^{r} (\alpha) =   \log(2) + \alpha \log\left( p_{r,K} \right)  = \phi_a^{r}(\alpha)\, ,
\end{align*}
and the RS capacity equals the annealed capacity eq.~\eqref{proof:first_moment}:
\begin{align*}
	\alpha_{\rm RS}^{r}(K) = \alpha_a^{r}(K) = \frac{-\log(2)}{\log\left( p_{r,K} \right)} \, .
\end{align*}

\subsubsection{$U$-function}
\begin{itemize}
		\item For $K  \leq K^*$, only the symmetric fixed point $(q_0,\hat{q}_0)= (0,0)$ exists, which leads again to the annealed free entropy:
		\begin{align*}
			\phi_{\rm RS}^{u} (\alpha) =   \log(2)
                                + \alpha \log\left( p_{u,K}
                               \right)  =  \phi_a^{u}(\alpha) \, ,
		\end{align*}
		and annealed capacity eq.~\eqref{proof:first_moment}:
		\begin{align*}
			\alpha_{\rm RS}^{u}(K) = \alpha_a^{u}(K) = \frac{-\log(2)}{\log\left( p_{u,K} \right)} \, .
		\end{align*}
		\item For $K >  K^*$, the RS entropy does not match
                  the annealed entropy because the fixed point
                  $(q_0,\hat{q}_0)\neq(0,0)$ corresponds to a lower
                  free energy than the symmetric fixed point
                  $(0,0)$. The symmetric fixed point becomes unstable for
                  $K>K^*$, where $K^*$ is remarkably given by the same value as in the independent section~\ref{proof:second_moment_u}. Hence it naturally verifies eq.~\eqref{AT_crossover_RS} even though its definition derives from the stability of the RS solution, that we study in the next section.
\end{itemize}

\subsection{Stability}
The local stability of the RS solution can be studied using de Almeida and Thouless (AT) method \textbf{\cite{22}}, based on the positivity of the Hessian of $S_n (\tbf{Q},\tbf{\hat{Q}} )$. The replica symmetric AT-line $\alpha_{\rm AT}$ is given by the solution of the following implicit equation (Appendix \ref{appendix:AT_stability}): 
\begin{equation*}
\frac{1}{\alpha} = \frac{1}{(1-q_0(\alpha))^2} \int Dt
\frac{\left(f_0^{z}(f_0^{z}-f_2^{z}) + (f_1^{z})^2
  \right)^2}{(f_0^{z})^4}(t,q_0(\alpha)) \int Dt \frac{
  \left(g_0^{w}g_2^{w} -(g_1^{w})^2  \right)^2
}{(g_0^{w})^4}(t,\hat{q}_0(\alpha))\,.
 \end{equation*}

As illustrated above, for the rectangle and $u-$function, the symmetry of the weights $P_w$ and the constraint $\varphi$ imposes the existence of the symmetric fixed point $(q_0,\hat{q}_0)=(0,0)$. This simplifies the previous condition and becomes equivalent to the linear stability condition of the symmetric fixed point $(q_0,\hat{q}_0)=(0,0)$ (see Appendix \ref{appendix:AT_stability}):  
\begin{equation*}
	\frac{1}{\alpha_{\rm AT}}=
        \left(\frac{\tilde{f}_2^{z}-\tilde{f}_0^{z}}{\tilde{f}_0^{z}}\right)^2
        \left(\frac{\tilde{g}_2^{w} }{\tilde{g}_0^{w}}\right)^2\, , \textrm{  where for $i\in \mathbbm{N}$: }
\begin{cases}
\tilde{g}_i^{w} = \displaystyle \int dw w^i P_w(w) e^{ \frac{w^2}{2}} \, ,\vspace{0.2cm}  \\
\tilde{f}_i^{z} = \displaystyle \int Dz z^i \varphi(z) \, .
\end{cases}
\end{equation*}

We plotted the annealed capacity, the replica symmetric capacity and
the AT-line for the step, rectangle and $u$-function binary
perceptrons as functions of $K$ in
fig.~\ref{main:plot_RS_capacity_step}, \ref{main:plot_RS_capacity_rectangle},
\ref{main:plot_RS_capacity_symstep}.
 
\subsubsection{Step binary perceptron}
We note that for the step binary perceptron the RS solution is always stable towards 1RSB, even for negative threshold $K<0$. This is interesting in the view of recent work on the spherical perceptron with negative threshold where the replica symmetry breaks for all $K<0$, and full-step RSB is needed to evaluate the storage capacity \cite{20}.

\subsubsection{Rectangle}
As the RS capacity $\alpha_{\rm RS}^{r}$ is always below the AT line $\alpha_{\rm AT}^{r}$, the RS solution is always locally stable.

\subsubsection{$u$-function }
There is a crossing between the values of the RS capacity $\alpha_{\rm RS}^{u}$ and the AT-line $\alpha_{\rm AT}^{u}$, which defines implicitly the value $K^*\simeq 0.817$, and matches the equality in eq.~\eqref{main:AT_second_moment}: 
\vspace{-0.3cm}
\begin{equation}
	 \frac{-\log\left(2 \right)}{ \log\left(p_{u,K^{*}}  \right)} =
         \frac{\pi}{2}\frac{\left( p_{u,K^{*}} \right)^2}{e^{-(K^*)^2}(K^*)^2}\, .
	 \label{main:AT_crossover_RS}
\end{equation}
For $K\leq K^*$, the RS solution is locally stable, while for $K>K^*$ the RS solution becomes unstable, and a symmetry breaking solution appears.

\begin{figure}[htb!]
\centering
\includegraphics[scale=0.32]{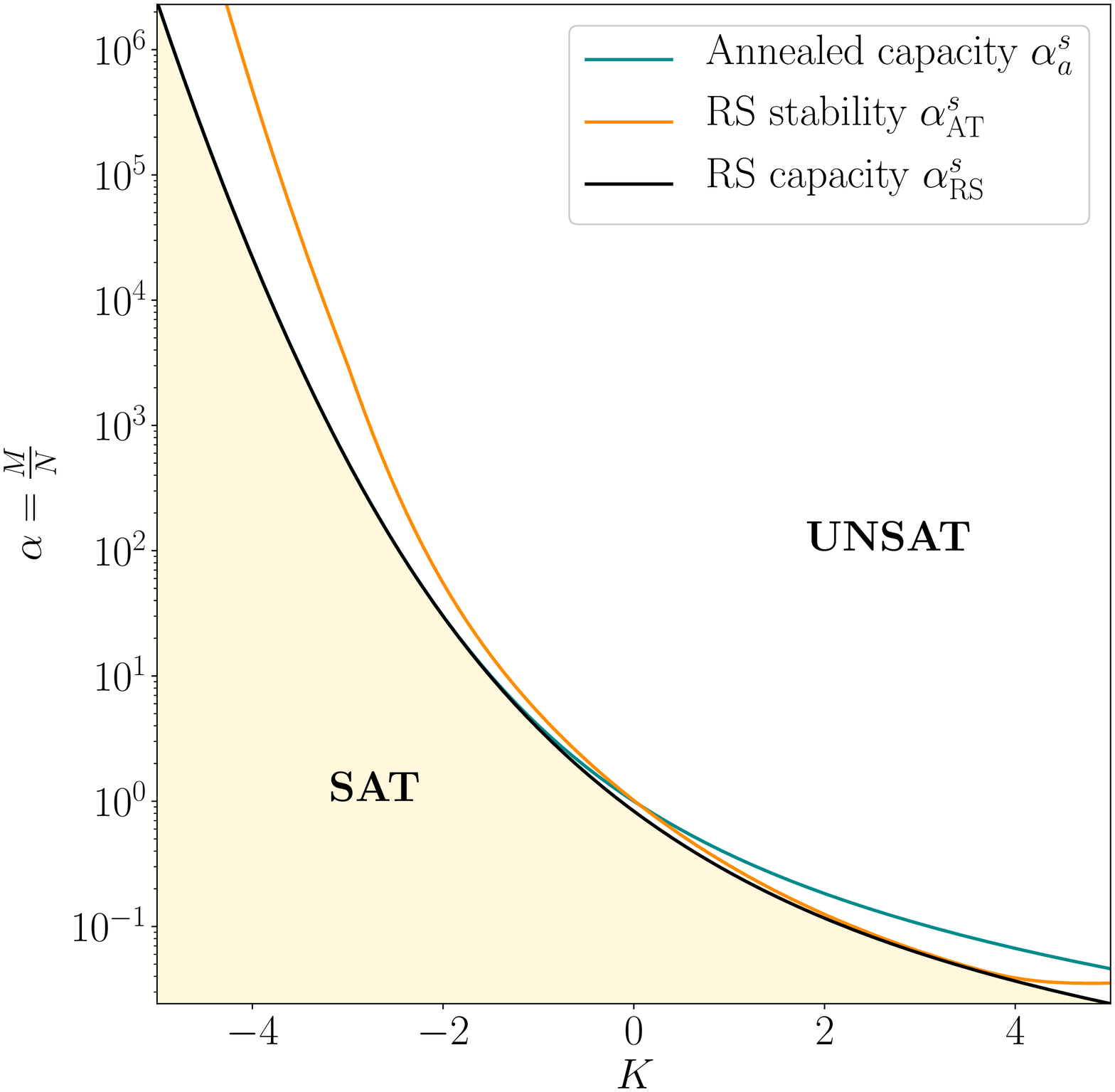}
\hspace{0.2cm}
\includegraphics[scale=0.32]{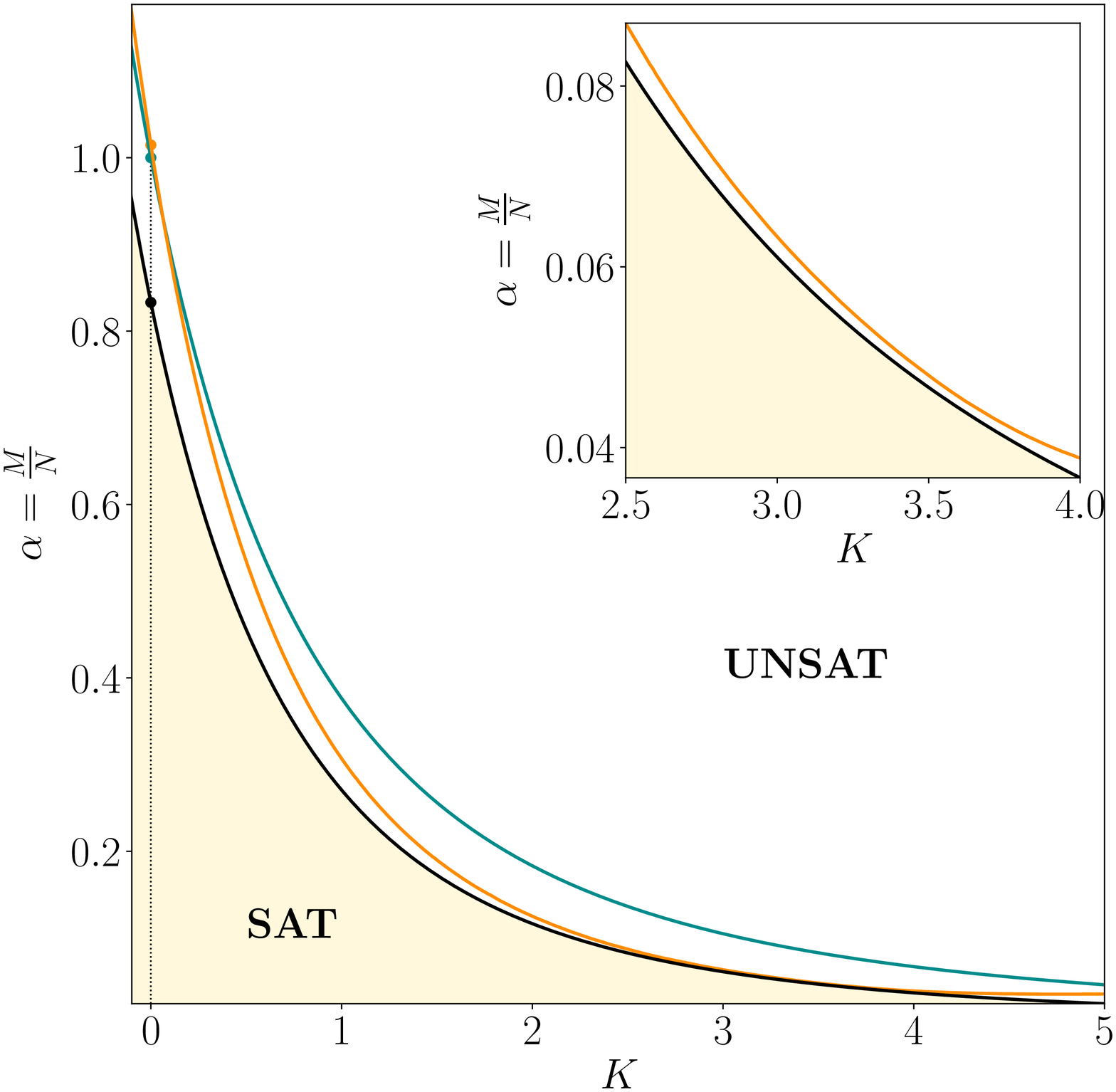}
\caption{Step binary perceptron (SBP): the RS capacity $\alpha_{\rm RS}^s$
  (black) does not match the annealed capacity $\alpha_a^s$ (blue) and
  is always below the AT-line $\alpha_{\rm AT}^s$ (orange). The AT-line is
  closest to the annealed capacity for $K_{\rm min} \simeq 3.62$ where the
  difference $\alpha_{\rm AT}^s - \alpha_a^s \simeq 0.0012$. For $K=0$, we retrieve well known results \cite{2}:
  $\alpha_{\rm RS}^r \simeq 0.833$, $\alpha_{\rm AT}^r \simeq 1.015$ and
  $\alpha_a^r = 1$. The left and right hand sides, and the inset, represent the same
  data on different scales. The satisfiable (SAT) phase is represented by the beige shaded area and is located below the RS capacity, while the unsatisfiable (UNSAT) starts at the capacity (black line) and extends for a larger number of constraints.}
	\label{main:plot_RS_capacity_step}
\end{figure}
\FloatBarrier

\begin{figure}[htb!]
\centering
\includegraphics[scale=0.32]{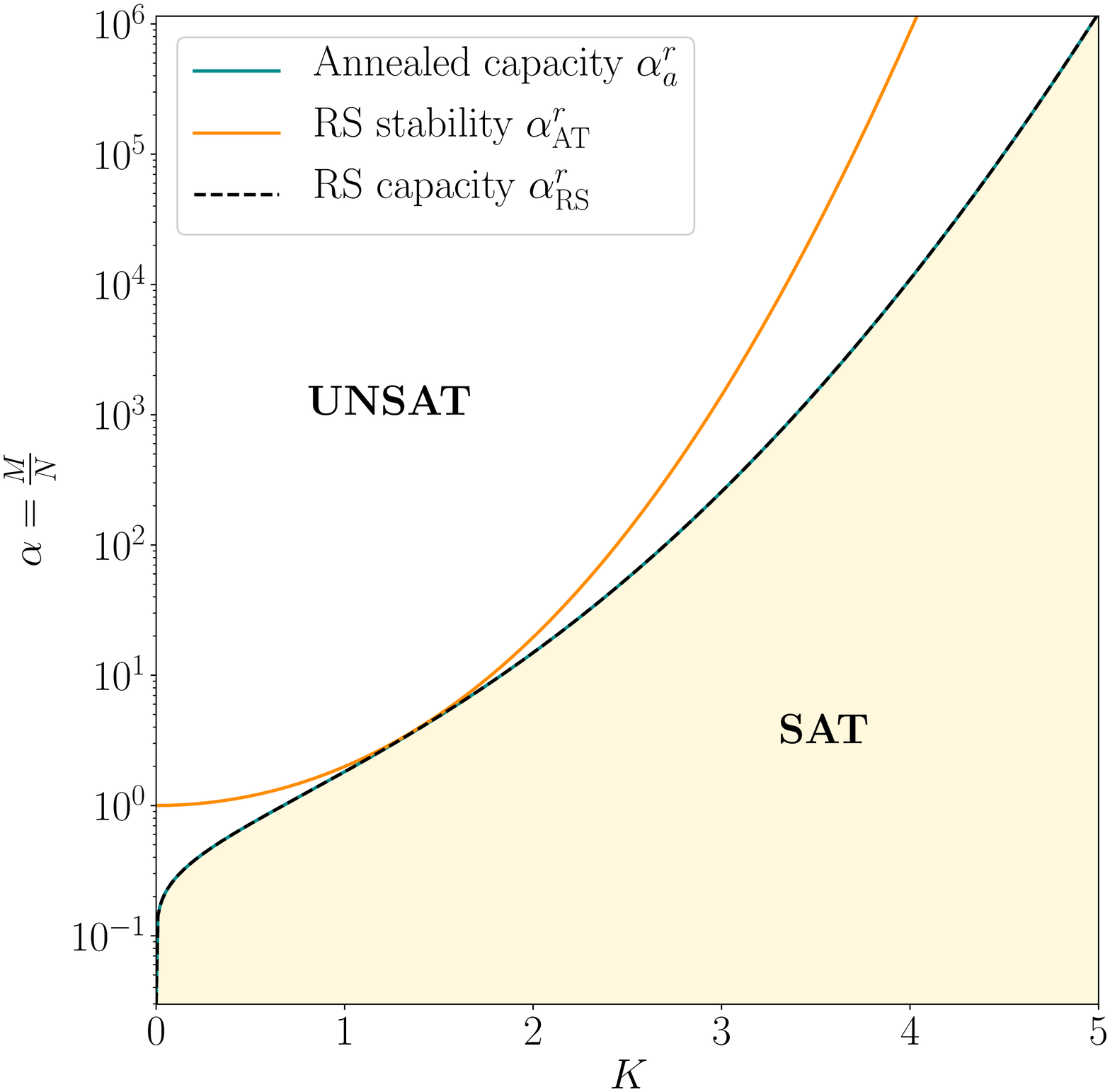}
\hspace{0.2cm}
\includegraphics[scale=0.32]{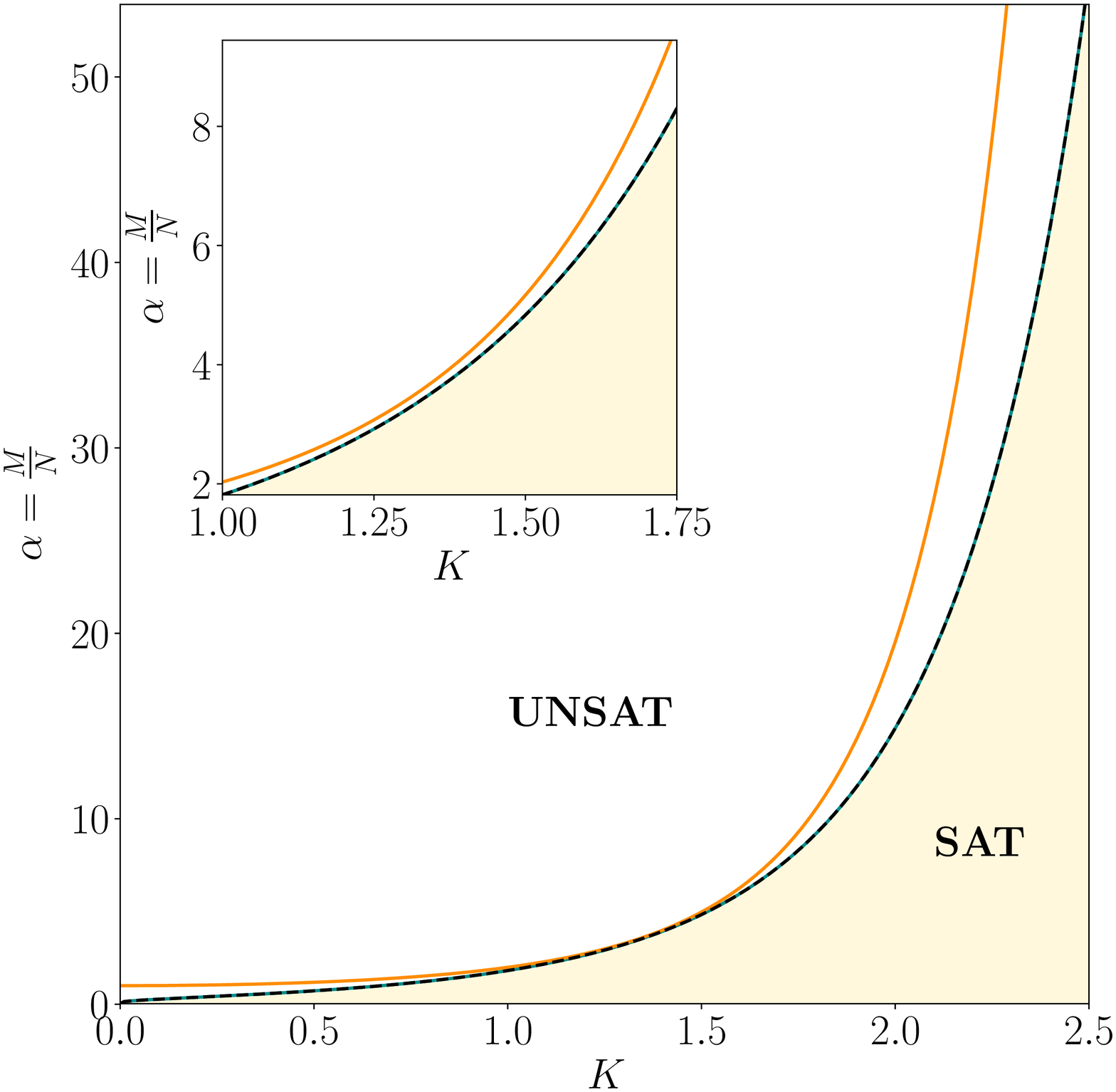}
\caption{Rectangle binary perceptron (RBP): the RS capacity
  $\alpha_{\rm RS}^{r}$ (black) matches the annealed bound $\alpha_{a}^{r}$ (blue), and
  the RS solution is locally stable for all $K$:
  $\alpha_{\rm RS}^{r}<\alpha_{\rm AT}^{r}$. The AT-line (orange) is
  closest to the annealed capacity for $K_{\rm min} \simeq 1.24$ where the
  difference $\alpha_{\rm AT}^s - \alpha_a^s \simeq 0.15$. The left and right hand sides, and the inset, represent the same
  data on different scales. The satisfiable (SAT) phase is represented by the beige shaded area and is located below the RS capacity, while the unsatisfiable (UNSAT) starts at the capacity (black line) and extends for a larger number of constraints.}
	\label{main:plot_RS_capacity_rectangle}
\end{figure}
\FloatBarrier

\begin{figure}[htb!]
\centering
\includegraphics[scale=0.32]{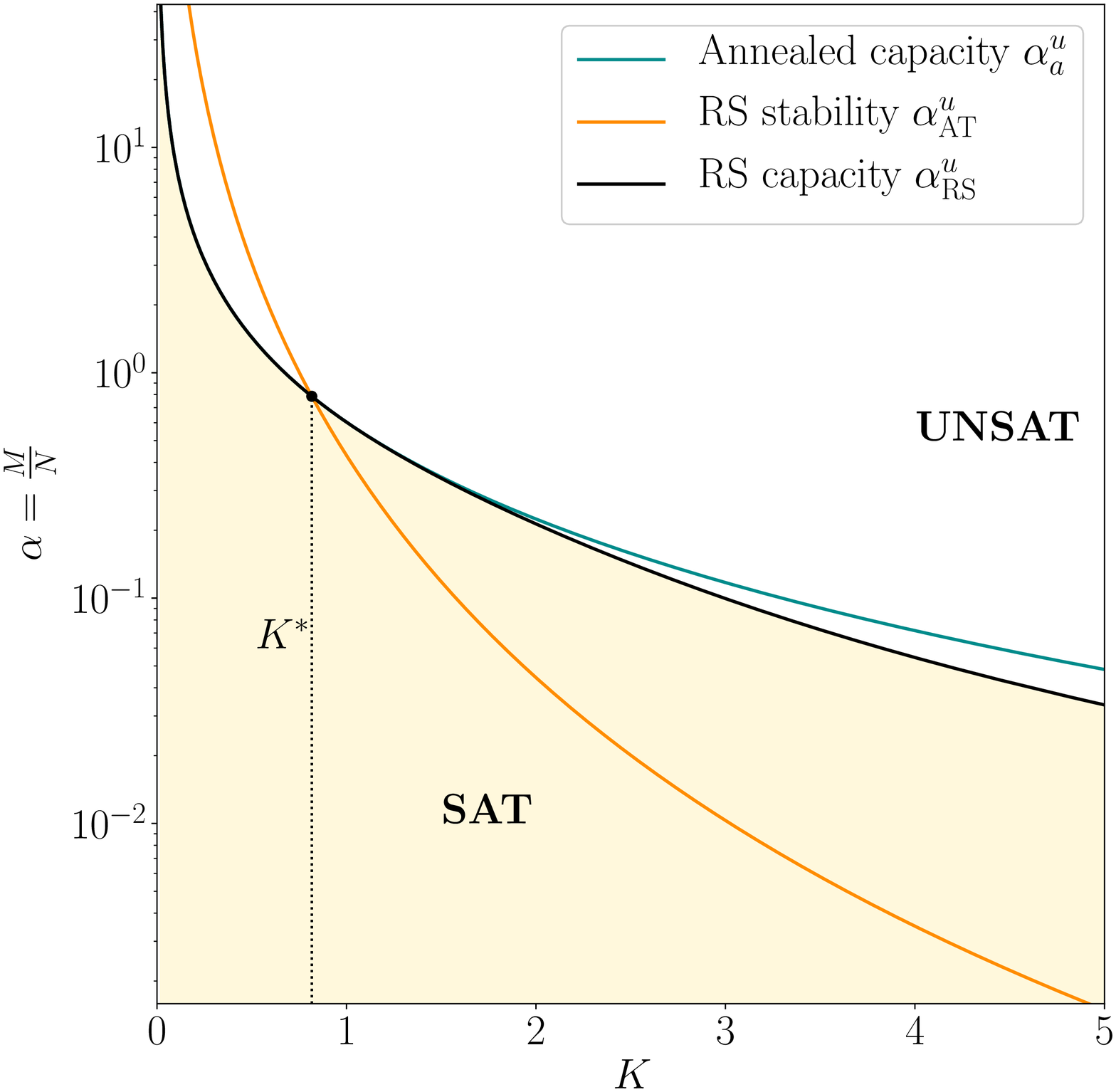}
\hspace{0.2cm}
\includegraphics[scale=0.32]{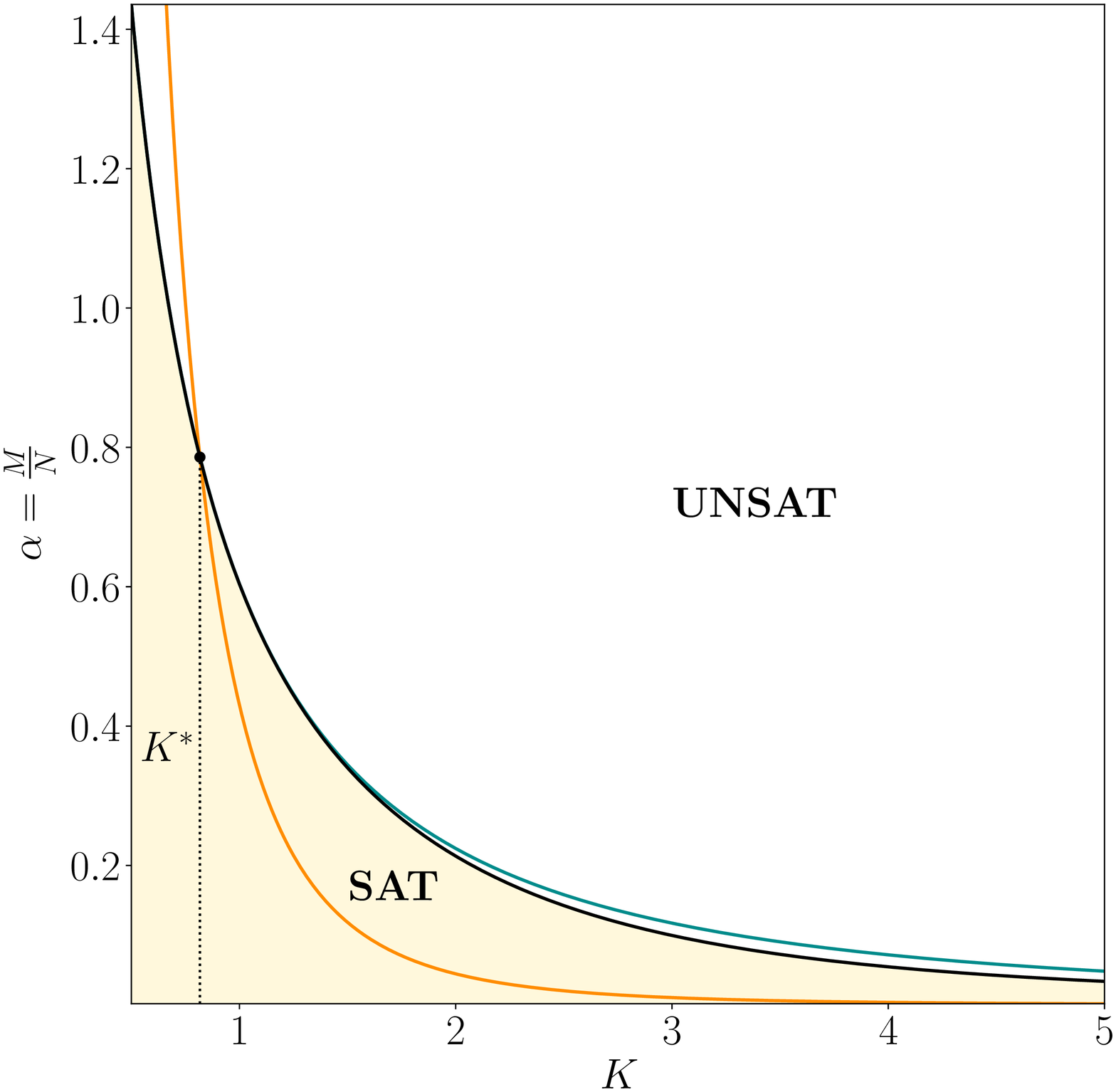}
\caption{$U-$function binary perceptron (UBP): the RS capacity §black) matches the
  annealed bound (blue) for $K<K^*$. At $K=K^*$, the RS capacity crosses the
  AT-line (orange). For $K>K^*$, the RS solution is unstable and the RS
  capacity deviates from the annealed capacity. The left and right hand sides, and the inset, represent the same
  data on different scales. The satisfiable (SAT) phase is represented by the beige shaded area and is located below the RS capacity, while the unsatisfiable (UNSAT) starts at the capacity (black line) and extends for a larger number of constraints.}
	\label{main:plot_RS_capacity_symstep}
\end{figure}
\FloatBarrier

\subsection{1RSB calculation}

In the previous section we concluded that the replica symmetric
solution is unstable in the $u-$function binary perceptron for $K>K^*$,
we analyze therefore the first-step of replica symmetry breaking (1RSB)
ansatz in this section. This ansatz and calculations is due to seminal
works of G. Parisi and is classic in the field of disordered systems and well presented in the
literature \cite{13,parisi1979infinite,parisi1980sequence,parisi1980order}, we thus mainly give the key formulas and defer the
details into the Appendix \ref{appendix:replica_1RSB}. 

The 1RSB ansatz assumes that the space of configurations
splits into states. Consequently replicas are not symmetric anymore and instead $n$ replicas are organized in $\frac{n}{m}$ groups containing $m$ replicas each:
\begin{equation}
\forall (a,b) \in [1:n]\times [1:n], \hspace{0.1cm} \frac{1}{N}(\tbf{w}^a\cdot\tbf{w}^b) = 
	\begin{cases}
		q_1  \textrm{ if $a$,$b$ belong to the same state,}\\
		q_0 \textrm{ if $a$,$b$ do not belong to the same state,}\\
		Q = 1 \textrm{ if } a = b \, .
	\end{cases}	
	\label{main:1RSB_ansatz}
\end{equation}

Following \textbf{\cite{25}}, the partition function $\mathcal{Z}_m$
associated to $m$ replicas falling in the same state is expressed as a sum over all possible
states $\Psi$ weighted by their corresponding free entropy $\phi$:
\begin{equation*}
	\mathcal{Z}_m = \sum_{\{\Psi\}} \exp(N m \phi(\Psi))
        =\sum_{\{\phi\}} {\cal N}_{\phi} \exp(N m
        \phi)=\sum_{\{\phi\}}  \exp(N \Sigma(\phi)) \exp(N m \phi)
        \sim \int d\phi \exp(N(m\phi+\Sigma(\phi))  \, ,
\end{equation*}
where we introduced the number of states at a given free
entropy $\phi$: ${\cal N}_{\phi}\equiv\exp(N \Sigma(\phi))$ and the
complexity $ \Sigma(\phi)$, also called the configurational entropy.

Using the saddle point method in the $N \to \infty$ limit, the 1RSB
replicated free entropy $\Phi_{\rm 1RSB}$ is written as a function of the Parisi parameter $m$, the free entropy $\phi$ and the complexity $\Sigma(\phi)$:
\begin{equation}
	\Phi_{\rm 1RSB}(m,\alpha) \equiv  \lim_{N \to \infty} \frac{1}{ N}
        \mathbb{E}_{\tbf{X}}\left[\log( \mathcal{Z}_m(\tbf{X}) )\right]  = m
        \phi + \Sigma(\phi) \, .
        \label{main:phi_1RSB_phi_m_Sigma}
\end{equation}

Injecting the 1RSB ansatz eq.~\eqref{main:1RSB_ansatz} in the replica
derivation eq.~\eqref{main:replica_general}, the 1RSB replicated free
entropy $\Phi_{\rm 1RSB}$ is written as a saddle point equation over
$\tbf{q}=(q_0,q_1)$ and $\tbf{\hat{q}}=(\hat{q}_0,\hat{q}_1)$ (see Appendix \ref{appendix:replica_1RSB}):
\begin{equation}
\Phi_{\rm 1RSB}(m,\alpha) =  \underset{\tbf{q} , \tbf{\hat{q}}}{\textrm{extr}} \left\{  \frac{m}{2} \left(  q_1\hat{q}_1 - 1 \right) + \frac{m^2}{2} \left(q_0\hat{q}_0 - q_1\hat{q}_1 \right)   + m \mathcal{I}^w_{\rm 1RSB}(\tbf{\hat{q}})   +\alpha m \mathcal{I}^{z}_{\rm 1RSB}(\tbf{q})    \right\}
\label{main:phi_1RSB}
\end{equation}

\begin{equation*}\textrm{ with }
	\begin{cases}
		\mathcal{I}^w_{\rm 1RSB}(\tbf{\hat{q}}) = \frac{1}{m}
                \int Dt_0 \log\left(\int Dt_1 g_0^w\left(\tbf{t},\tbf{\hat{q}} \right )^{m}
                \right)\, ,
		 \vspace{0.3cm} \\
		\mathcal{I}^{z}_{\rm 1RSB}(\tbf{q}) = \frac{1}{m}
                \int Dt_0  \log\left(  \int Dt_1  f_0^z\left(\tbf{t},\tbf{q}\right)^m
                    \right) \, ,
	\end{cases}
\end{equation*}

\begin{equation}\textrm{denoting $\tbf{t}=(t_0,t_1)$, and for $i\in \mathbbm{N}$: }
	\begin{cases}
		g_i^w (\tbf{t},\tbf{\hat{q}})  = \int dw \, w^i P_w(w) \exp\left(
                  \frac{(1- \hat{q}_1  )}{2} w^2 +
                  \left(\sqrt{\hat{q}_0}t_0+\sqrt{\hat{q}_1-\hat{q}_0}t_1
                  \right)w \right)\, ,
			 \vspace{0.3cm} \\
		f_i^z(\tbf{t},\tbf{q}) = \int Dz  \, z^i \varphi(\sqrt{q_0} t_0 +
                \sqrt{q_1-q_0} t_1  + \sqrt{1-q_1} z)\, .
 	\end{cases}
 	\label{main:f_i_g_i}
\end{equation}

Taking the derivative of $\Phi_{\rm 1RSB}$ with respect to $m$, the free
entropy $\phi$ and complexity $\Sigma$ can be written as: 
\begin{equation}
	\begin{cases}
		\phi(\alpha) = \frac{\partial \Phi_{\rm 1RSB}(m,\alpha)  }{\partial m }  = \underset{\tbf{q} , \tbf{\hat{q}}}{\textrm{extr}} \left\{ \frac{1}{2} (  q_1\hat{q}_1 - 1 ) +   m \left(q_0\hat{q}_0 - q_1\hat{q}_1 \right)
	+\mathcal{J}^w_{\rm 1RSB}(\tbf{\hat{q}} ) + \alpha
        \mathcal{J}^{z}_{\rm 1RSB}(\tbf{q}) \right\} \, ,\vspace{0.3cm} \\
	\Sigma (\phi) =  \Phi_{\rm 1RSB}- m \phi = \underset{\tbf{q} ,
          \tbf{\hat{q}}}{\textrm{extr}} \left\{
          \frac{m^2}{2}(q_1\hat{q}_1 - q_0\hat{q}_0) +
          m(\mathcal{I}^w_{\rm 1RSB} - \mathcal{J}^w_{\rm 1RSB}
          )(\tbf{\hat{q}}) + m \alpha (\mathcal{I}^{z}_{\rm 1RSB}
          -\mathcal{J}^{z}_{\rm 1RSB})(\tbf{q})\right\}\, ,
	\end{cases}
	\label{main:complexity}
\end{equation}

\begin{equation*} \textrm{ with }
	\begin{cases}
		\mathcal{J}^w_{\rm 1RSB}(\tbf{\hat{q}}) =
                \frac{\partial \left(m \mathcal{I}^w_{\rm
                      1RSB}\right)}{\partial m} = \int Dt_0 \frac{\int
                  Dt_1 \log\left(g_0^w(\tbf{t},\tbf{\hat{q}})\right)
                  g_0^w(\tbf{t},\tbf{\hat{q}})^m}{\int Dt_1
                  g_0^w(\tbf{t},\tbf{\hat{q}})^m} \, ,\vspace{0.2cm} \\ 
		\mathcal{J}^{z}_{\rm 1RSB}(\tbf{q}) = \frac{\partial
                  \left(m \mathcal{I}^{z}_{\rm 1RSB}\right) }{\partial
                  m} = \int Dt_0 \frac{\int Dt_1
                  \log\left(f_0^z(\tbf{t},\tbf{q})\right) f_0^z(\tbf{t},\tbf{q})^m}{\int
                  Dt_1 f_0^z(\tbf{t},\tbf{q})^m} \, .\vspace{0.2cm} \\ 
	\end{cases}
\end{equation*}

\subsection{1RSB results for UBP}
From now on, we only consider the $u-$function binary perceptron, whose RS solution is unstable for $K>K^*$.
To describe the equilibrium of the system in the SAT phase, we need to find the value of the Parisi parameter at equilibrium $m_{\rm eq}$. The complexity $\Sigma(\phi)$ is the entropy of clusters having internal entropy $\phi$. In order to capture clusters that carry almost all configurations, we need to maximize the total entropy $\phi_{\rm tot}  = \Sigma(\phi) + \phi$ under the constraint that the free entropy and complexity are both positive $\phi \geq 0$ and $\Sigma(\phi)\geq 0$. Hence from eq.~\eqref{main:phi_1RSB_phi_m_Sigma}, the equilibrium Parisi parameter $m_{\rm eq}$ verifies 
\begin{equation*}
	\phi_{\rm eq}=\underset{\phi \geq 0,\Sigma \geq
          0}{\textrm{argmax }}{ \{ \phi + \Sigma(\phi) \} } \hspace{1cm} \textrm{and} \hspace{1cm}
    m_{\rm eq} =    \left. - \frac{d \Sigma}{d \phi} \right |_{\phi_{\rm eq}}\, .
\end{equation*}



Using the expressions eq.~\eqref{main:complexity} and varying the Parisi
parameter $m\in[0;1]$, we obtain the curve of the complexity
$\Sigma(\phi)$ as shown in fig.~\ref{main:complexity_curves}. At
$m=1$, the complexity is negative. Decreasing $m$, the complexity
increases and becomes positive at the value $m_{\rm eq}$. Besides for
small values of $m$, an unphysical (convex) branch appears, as commonly
observed in other systems solved by the replica method.

We note that at $\alpha$ increases both the equilibrium complexity and
free entropy decrease. In constraint satisfaction problems such as
K-satisfiability or random graph coloring the mechanism in which the
satisfiability threshold appears is that the maximum of the complexity
becomes negative. In the present UBP problem it is actually both the
free entropy and the complexity that vanish together, as illustrated
in fig.~\ref{main:complexity_curves}.

\begin{figure}
\centering
\includegraphics[scale=0.30]{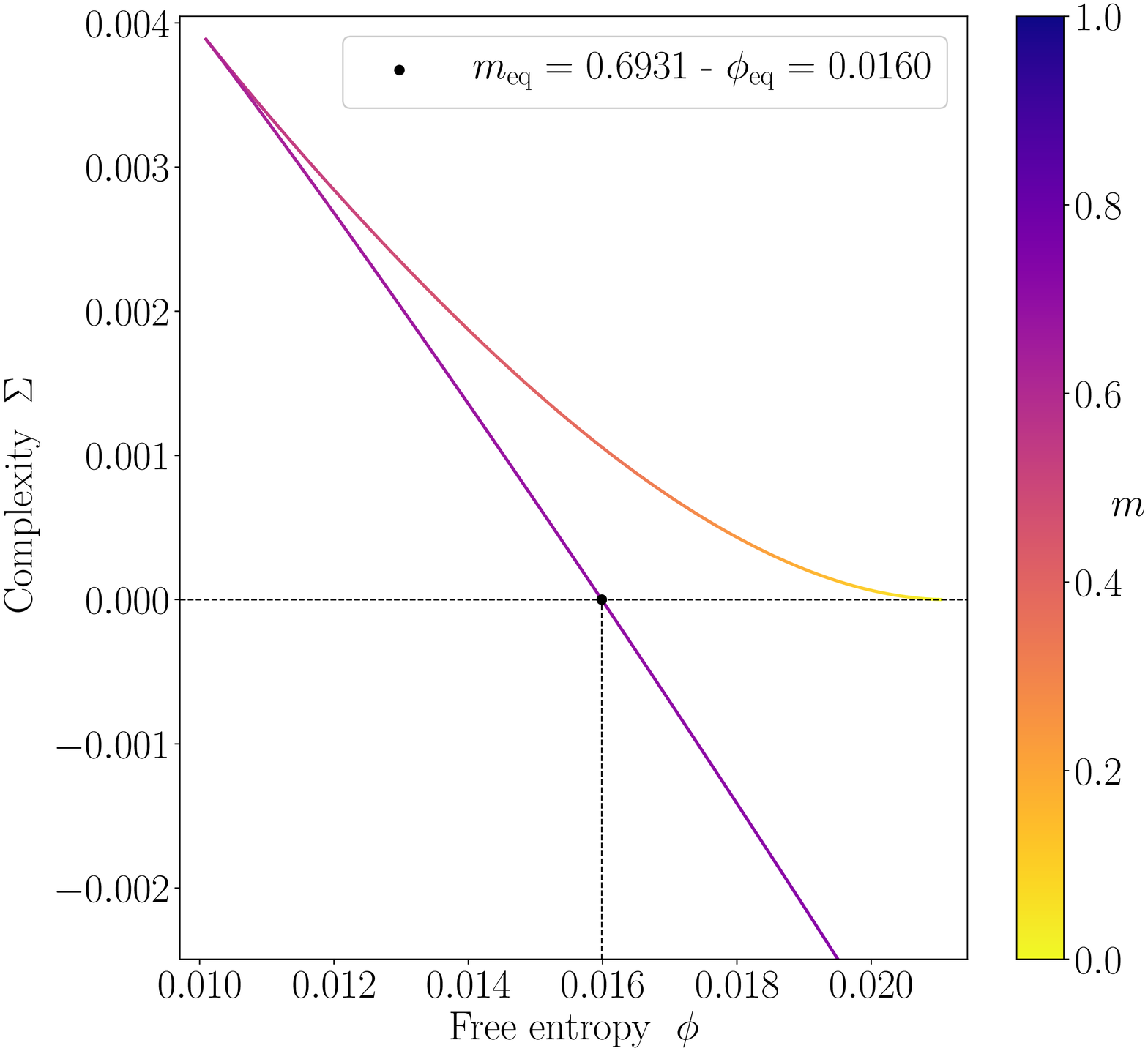}
	\hspace{0.2cm}
	\includegraphics[scale=0.310]{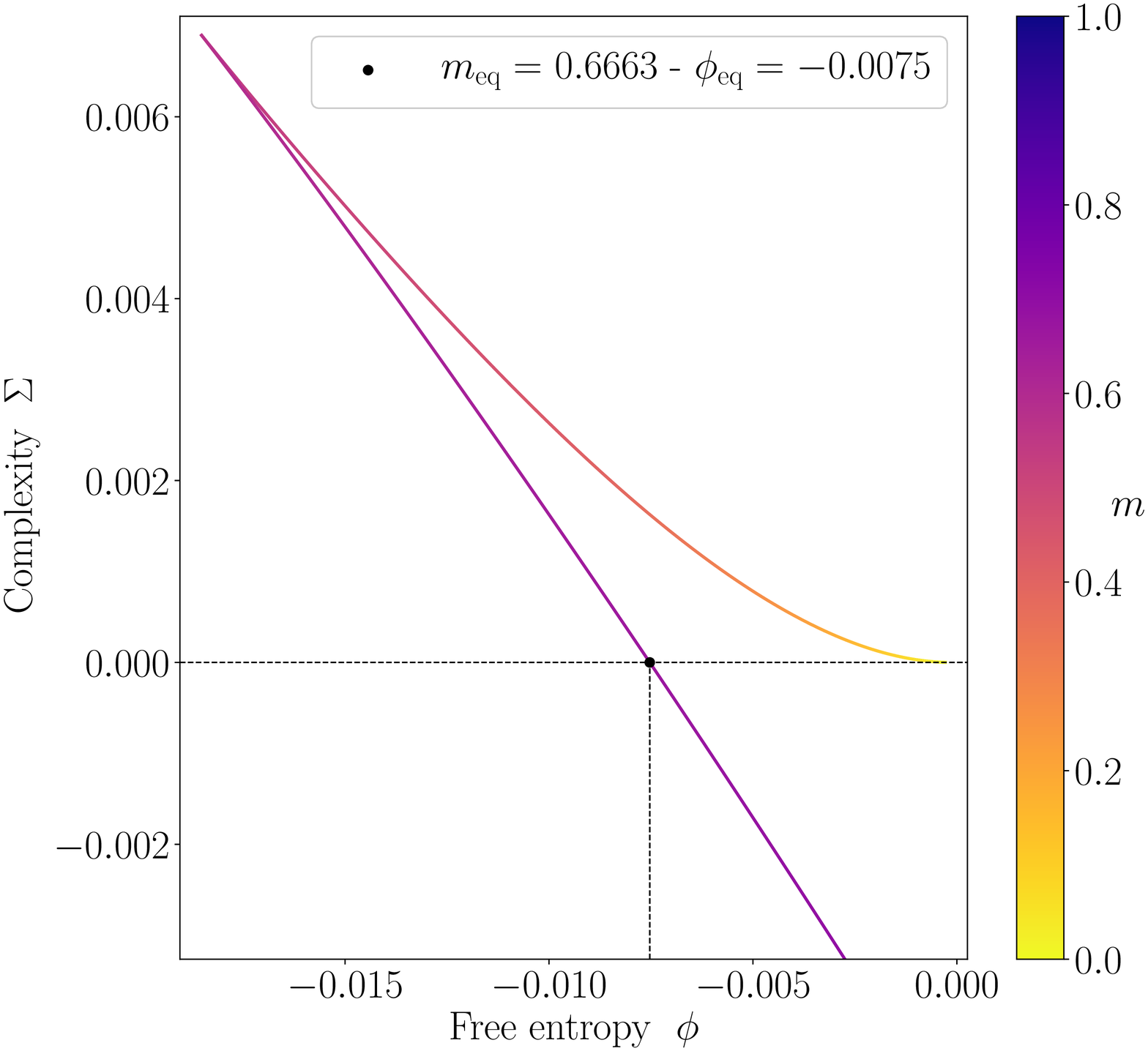}
	\caption{Complexity $\Sigma(\phi)$ as a function of the free
          entropy $\phi$ for the $u-$function binary perceptron  at
          $K=1.5>K^*$. Complexity reaches $\Sigma=0$ (black dot) at
          $m_{\rm eq}$. For $K=1.5$ and $\alpha=0.33$ \textbf{a)} the free-entropy
          corresponding to $m_{\rm eq}$ is positive $\phi_{\rm eq} >0$, whereas for
          $\alpha=0.34$ \textbf{b)} the free entropy at $m_{\rm eq}$ is negative $\phi_{\rm eq} < 0$ and therefore there is no part of the curve where both
          complexity and free entropy are positive: thus this value
          of $\alpha$ is beyond the 1RSB storage capacity, and the capacity is in the interval $[0.33;0.34]$.}
	\label{main:complexity_curves}
	\end{figure} 
\FloatBarrier

Computing the equilibrium value $m_{\rm eq}(\alpha)$, we have access to
the corresponding equilibrium overlaps $q_0^*$ and $q_1^*$, that we
may compare with the RS solution $q_{\rm RS}$. All these are depicted
in fig.~\ref{main:plot_1RSBmeq}. 
The function $m_{\rm eq} ( \alpha )$ shows a non monotonic behaviour as it has been previously observed, e.g. in the Sherrington-Kirkpatrick model as a function of temperature \cite{Mezard1987}.

 We also compute the 1RSB entropy that verifies $\phi_{\rm 1RSB}^{u} \leq \phi_{\rm RS}^{u}$ and which vanishes at the 1RSB capacity $\alpha_{\rm 1RSB}^{u}$ as
depicted in fig.~\ref{main:plot_annealed_minus_replicas}a. We note that the above inequality is as predicted by Parisi's replica theory \cite{Mezard1987}, taking into account that we are working at strictly zero energy, where the entropy becomes minus the free energy. 

The 1RSB solution provides a small correction to the RS result for storage capacity, as illustrated in
fig.~\ref{main:plot_annealed_minus_replicas}b, where we plotted the
difference between the annealed upper bound and the capacity for the RS
and 1RSB solutions: $\alpha_a^{u}-\alpha_{\rm RS}^{u}$ and
$\alpha_a^{u}-\alpha_{\rm 1RSB}^u$. 

\begin{figure}[htb!]
\centering
	\includegraphics[scale=0.32]{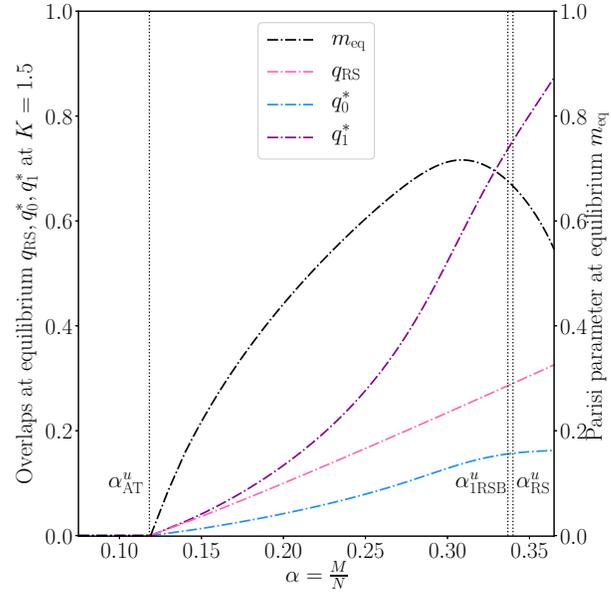}
	\caption{Equilibrium values of the overlap
          $q_0^*\ne q_{\rm RS}$, $q_1^*$ and the Parisi parameter $m_{\rm eq}$ for the UBP at
          $K=1.5$. For $K<K^*$, the RS solution is stable and the only fixed point is $q_0=q_1=q_{RS}=0$. }
	\label{main:plot_1RSBmeq}
	\label{main:plot_entropy_1RSB}
\end{figure} 
\FloatBarrier

\begin{figure}[htb!]
		    \centering
		    \includegraphics[scale=0.31]{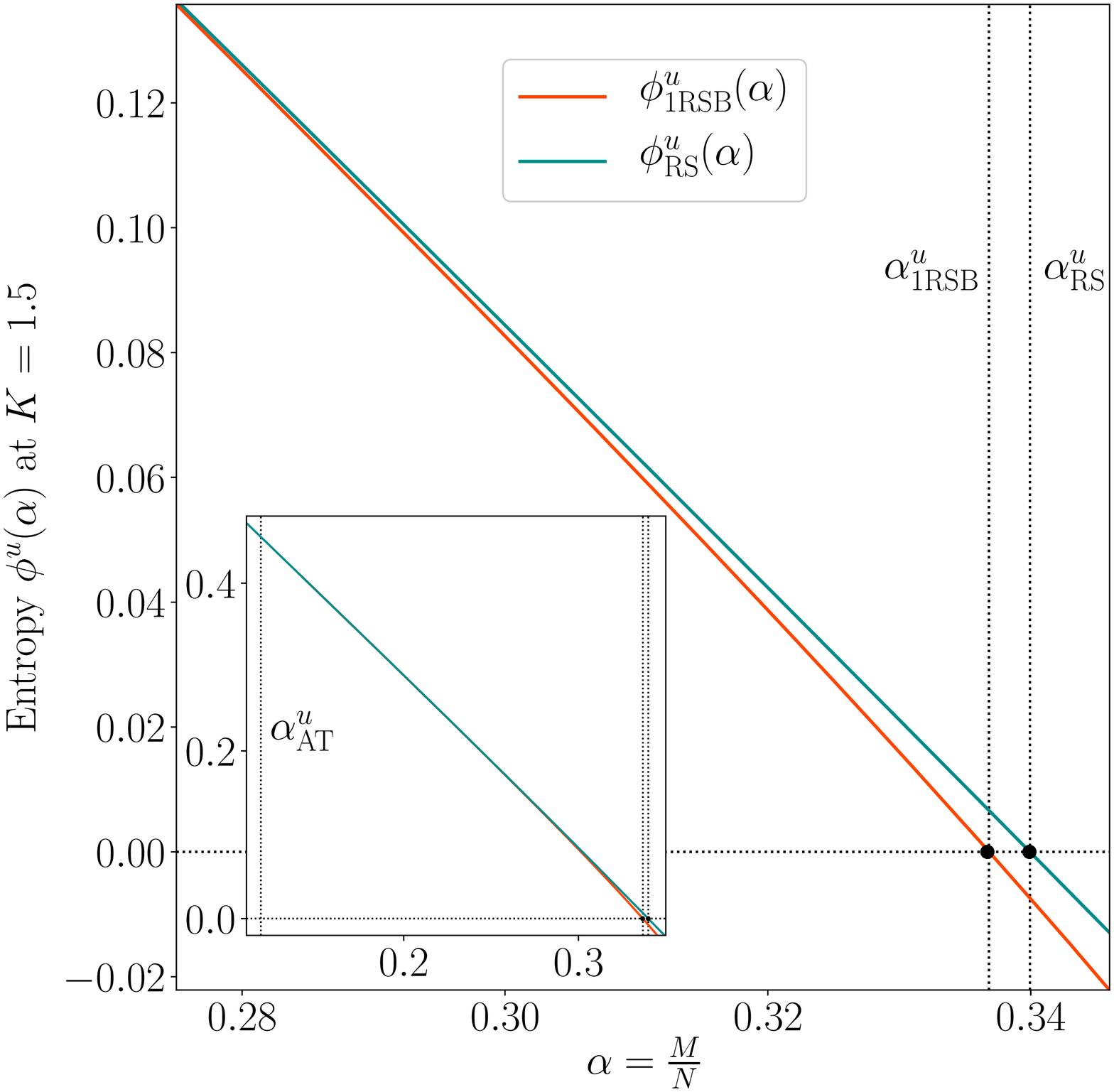}
		    \hspace{0.2cm}
	   		\includegraphics[scale=0.31]{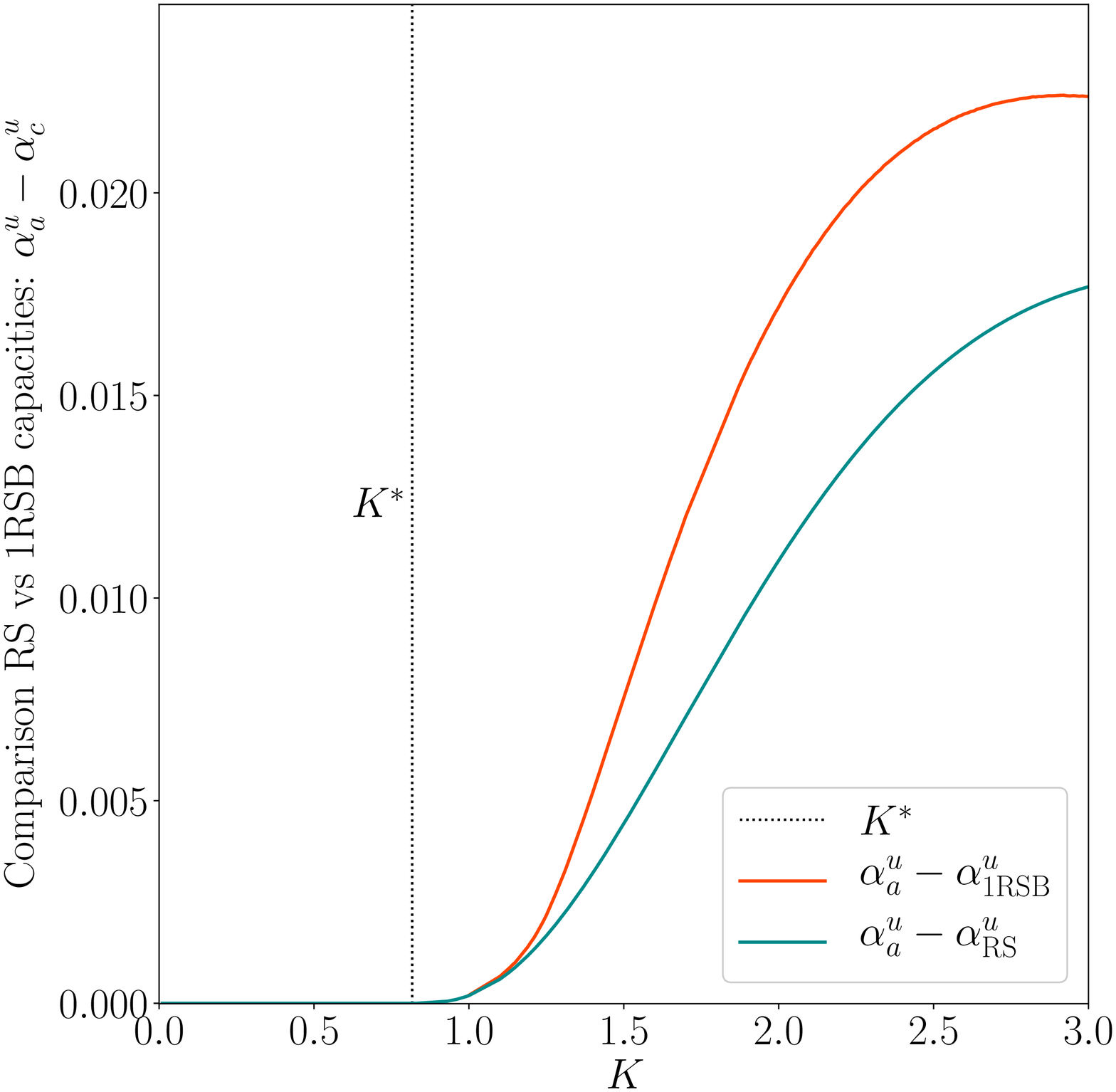}
	   		\caption{\textbf{a)}: Comparison of the RS (blue) and 1RSB (orange) entropy
          for the UBP at $K=1.5$. For $\alpha<\alpha_{\rm AT}\simeq 0.118 $, RS and
          1RSB entropies are equalled. For $\alpha>\alpha_{\rm AT}$,
          1RSB entropy deviates slightly of the RS entropy before vanishing respectively at $\alpha_{\rm 1RSB}^{u}\simeq 0.337 $ and $\alpha_{\rm RS}^{u}\simeq 0.334$. The inset represents the same data on a different scale. \textbf{b)}: Difference between the annealed upper
                          bound and the 1RSB capacity
                          $\alpha_a^{u}-\alpha_{\rm 1RSB}^{u}$ (orange) and the RS capacity
                          $\alpha_a^{u}-\alpha_{\rm RS}^{u}$ (blue). Below $K^*$ the RS solution is stable: RS and 1RSB entropies match exactly. Above $K^*$, the RS solution is unstable: the 1RSB entropy deviates slightly from the RS solution.}
	   		\label{main:plot_annealed_minus_replicas}
\end{figure}
\FloatBarrier

\subsection{1RSB Stability}

In the previous section we evaluated the 1RSB storage capacity of the
$u-$function binary perceptron for $K>K^*$. In this section we will argue that this
cannot be an exact solution to the problem. 

We could investigate the stability of 1RSB towards further levels of
replica symmetry breaking along the same lines we did for the RS
solution. However, in the present case we do not need to do that to
see that the obtained solution cannot be correct. The explanations
lies in the breaking of the up-down symmetry in the problem. This
symmetry must either be broken explicitly as in the ferromagnet,
where the system would acquire an overall magnetization, but we have
not observed any trace of this in the present problem. Or this up-down 
symmetry must be conserved in the final correct solution. The
conservation of the up-down symmetry is manifested in the value
$q_0=0$ in the replica symmetric phase. The fact that in the 1RSB
solution evaluated above we do not observe $q_0=0$, but instead
$q_0>0$ is a sign of the fact that we are evaluating a wrong
solution. The only possible way to obtain an exact solution we foresee
is to evaluate the full-step replica symmetry breaking with a
continuity of overlaps $q(x)$, the smallest one of them should be $0$
in order to restore the up-down symmetry.  We let the evaluation of
the full-RSB for future work. 
 
Finally let us note that the 1RSB solution obtained in the previous
section can be interpreted as frozen-2RSB. In 2RSB we would have 3
kinds of overlaps, $q_0$, $q_1$ and $q_2$. In frozen 2RSB we would have
$q_2=1$, $q_1=q_1^{\rm 1RSB}$, $q_0=q_0^{\rm 1RSB}$. 

\section{ Conclusion}

The step-function binary perceptron has thus far eluded a
rigorous establishment of the conjectured storage capacity, eq.~\eqref{RS_capacity}. This prediction
is expected to be exact because of the frozen-1RSB nature of the
problem \cite{2,16}. At the same time the work of
\cite{baldassi2015subdominant} sheds light on the fact that the
structure of the space of solutions is not fully described by the
frozen-1RSB picture, and that rare dense and unfrozen regions exist and
in fact are amenable to dynamical procedures searching for solutions. It
remains to be understood how is it possible that the 1RSB calculation
does not capture these dense unfrozen regions of solutions
\cite{baldassi2015subdominant}. 
They do not dominate the equilibrium, but the RSB calculation is
expected to describe rare events via their large deviations, which in
this case it does not. 

In this paper we focus on two cases of the binary perceptron with
symmetric constraints, the rectangle binary perceptron and the
$u-$function binary perceptron. We prove (up to a numerical assumption) using the second moment method that the
storage capacity agrees in those cases with the annealed upper bound,
except for the $u-$function binary perceptron for $K>K^*$ eq.~\eqref{AT_crossover_RS}. 
We analyze the 1RSB solution in that case and indeed obtain a lower
prediction for the storage capacity. However, we do not expect the
1RSB to provide the exact solution because it does not respect the
up-down symmetry of the problem. Though the precise nature of the satisfiable phase for the $u-$function binary
perceptron for $K>K^*$ remains illusive, we can conjecture it is
full-RSB \cite{parisi1979infinite,parisi1980sequence,parisi1980order}. Establishing this rigorously would provide much deeper
understanding and remains a challenging subject for future work. 

\section*{Acknowledgement}

We thank Florent Krzakala, Joe Neeman, and Pierfrancesco Urbani for useful
discussions. 
We acknowledge funding from the ERC under the European
Unionâ€s Horizon 2020 Research and Innovation Programme Grant
Agreement 714608-SMiLe.   WP was supported in part by EPSRC grant EP/P009913/1.

\bibliographystyle{unsrt}
\bibliography{Bibli.bib}

\section{Appendices}
%

\subsection{General replica calculation}
\label{appendix:general_replica_calculation}
We present here the replica computation for general prior distribution $P_w$ and constraint function $\varphi$. In order to compute the quenched average of the free entropy, we consider the partition function of $n \in \mathbb{N}$ identical copies of the initial system. Using the replica trick, and an analytical continuation, the averaged free entropy $\phi$ of the initial system  reads:
\begin{align}
	\phi(\alpha ) &\equiv  \lim_{N\rightarrow +\infty} \frac{1}{N} \mathbb{E}_{\tbf{X}}[\log(\mathcal{Z(\tbf{X})})] = \lim_{N\rightarrow +\infty} \lim_{n\rightarrow 0} \frac{1}{N}\frac{ \partial \log \left(\mathbb{E}_{\tbf{X}}[\mathcal{Z}(\tbf{X})^n]\right)}{\partial n}\, ,
\label{appendix:free_entropy}
\end{align} 
where the replicated partition function can be written as 
\begin{align}
\mathbb{E}_{\tbf{X}}[\mathcal{Z}(\tbf{X})^n] &=  \int d\tbf{X}P_X(\tbf{X}) \mathcal{Z}(\tbf{X})^n = \int d\tbf{X}P_X(\tbf{X})  \prod_{a=1}^n  \int d\tbf{w}^a P_w(\tbf{w}^a)\int d\tbf{z}^a  \mathcal{C}(\tbf{z}^a)\delta(\tbf{z}^a-\tbf{X}\tbf{w}^a)\,,
\label{appendix:average_Zn}
\end{align}
with the global constraint function $\mathcal{C}(\tbf{z}) = \displaystyle \prod_{\mu=1}^M \varphi(z_{\mu})$.

 
We suppose that inputs are $iid$ distributed from $P_X \triangleq \mathcal{N}\left(0,\frac{1}{N}\right)$. More precisely, for $i,j\in [1:N]$, $\mu,\nu \in [1:M]$, $\mathbb{E}_\tbf{X}[X_{i\mu} X_{j \nu}] = \frac{1}{N} \delta_{\mu\nu} \delta_{ij}$. Hence $z_{\mu}^a =\sum_{i=1}^N X_{i\mu}w_i^a$ is the sum of $iid$ random variables. The central limit theorem insures that $z_{\mu}^a \sim \mathcal{N}\left(\mathbb{E}_{\tbf{X}}[z_\mu^a]  ,\mathbb{E}_{\tbf{X}}[z_\mu^a z_\mu^b] \right)$, with two first moments:
\begin{equation}
	\begin{cases}
		\mathbb{E}_{\tbf{X}}[z_\mu^a] =\sum_{i=1}^N \mathbb{E}_{\tbf{X}}[X_{i\mu}] w_i^a =0\\
		\mathbb{E}_{\tbf{X}}[z_\mu^a z_\mu^b] = \sum_{ij} \mathbb{E}_{\tbf{X}}[X_{i\mu}X_{j\mu}] w_i^a w_j^b  = \frac{1}{N}\sum_{ij}  \delta_{ij} w_i^a w_j^b=\frac{1}{N} \sum_{i=1}^N w_i^a w_i^b  \, .
	\end{cases}
\end{equation}

In the following we introduce the symmetric overlap matrix $\tbf{Q}\equiv(\frac{1}{N} \sum_{i=1}^N w_i^a w_i^b)_{a,b=1..n}$. Define $\tbf{\tilde{z}}_{\mu} \equiv (z^a_{\mu})_{a=1..n}$ and $\tbf{\tilde{w}}_i \equiv (w_i^a)_{a=1..n}$.  $\tbf{\tilde{z}}_{\mu}$ follows a multivariate gaussian distribution $\tbf{\tilde{z}}_{\mu} \sim P_{\tilde{z}} \triangleq \mathcal{N}( \tbf{0}, \tbf{Q})$ and $
        P_{\tilde{w}}(\tbf{\tilde{w}}) = \prod_{a=1}^n
        [\delta(\tilde{w}_a-1) + \delta(\tilde{w}_a+1)]$. Introducing the change of variable and the Fourier representation of the $\delta$-Dirac function that involves a new parameter $\tbf{\hat{Q}}$:
$$1 = \int d\tbf{Q} \prod_{a \leq b} \delta \left(NQ_{ab}-\sum_{i=1}^N w_i^a w_i^b \right) = \int d\tbf{Q} \int d\tbf{\hat{Q}} \exp \left(-\frac{N}{2}\\Tr(\tbf{Q\hat{Q}} )  \right)  \exp\left( \frac{1}{2}\sum_{i=1}^N \tbf{\tilde{w}}_i^{\intercal} \tbf{\hat{Q}} \tbf{\tilde{w}}_i\right),$$
the replicated partition function becomes an integral over the matrix parameters $\tbf{Q}$ and $\tbf{\hat{Q}}$, that can be evaluated using Laplace method in the $N \to \infty$ limit,
\begin{align}
	\mathbb{E}_{\tbf{X}}[\mathcal{Z}(\tbf{X})^n] &= \int d \tbf{Q} d\tbf{\hat{Q}} e^{-N \left( \frac{1}{2}\Tr(\tbf{Q\hat{Q}} ) - \log\left(\int d\tbf{\tilde{w}} P_{\tilde{w}}(\tbf{\tilde{w}})  e^{ \frac{1}{2}\tbf{\tilde{w}}^{\intercal} \tbf{\hat{Q}} \tbf{\tilde{w}} }  \right) - \alpha \log\left( \int  d\tbf{\tilde{z}} P_{\tilde{z}}(\tbf{\tilde{z}}) \mathcal{C}(\tbf{\tilde{z}}) \right)  \right) }\\
	&= \int d \tbf{Q} d\tbf{\hat{Q}} e^{-N S_n (\tbf{Q},\tbf{\hat{Q}} ) } \underset{N \to \infty}{\simeq} e^{-N \cdot  \text{SP}_{\tbf{Q}, \tbf{\hat{Q}}} \left\{ S_n(\tbf{Q},\tbf{\hat{Q}}) \right\}}, 
	\label{appendix:expectation_Zn}
\end{align}
where $\textrm{SP}$ states for saddle point and we defined
\begin{equation}
     \begin{cases}
     S_n (\tbf{Q},\tbf{\hat{Q}}) = \frac{1}{2}\Tr(\tbf{Q\hat{Q}}) -\log(\mathcal{I}^w_n (\tbf{\hat{Q}}))-\alpha\log \left(\mathcal{I}^z_n(\tbf{Q}) \right)
      \vspace{0.2cm} \\
        \mathcal{I}^w_n (\tbf{\hat{Q}}) =  \int_{\mathbb{R}^n} d\tbf{\tilde{w}} P_{\tilde{w}}(\tbf{\tilde{w}})  e^{ \frac{1}{2}\tbf{\tilde{w}}^{\intercal} \tbf{\hat{Q}} \tbf{\tilde{w}} }    \vspace{0.2cm} \\
        \mathcal{I}^z_n(\tbf{Q}) =   \int_{\mathbb{R}^n}  d\tbf{\tilde{z}} P_{\tilde{z}}(\tbf{\tilde{z}}) \mathcal{C}(\tbf{\tilde{z}}).
     \end{cases}
     \label{appendix:S_n}
\end{equation}

Finally, using eq.~\eqref{appendix:free_entropy} and switching the two limits $n\to 0$ and $N \to \infty$, the quenched free entropy $\phi$ simplifies as a saddle point equation
\begin{equation}
\phi (\alpha) = - \text{SP}_{ \tbf{Q}, \tbf{\hat{Q}} } \left\{\lim_{n\rightarrow 0} \frac{\partial S_n(\bold{Q},\bold{\hat{Q}})}{\partial  n} \right\} ,
\label{appendix:free_entropy2}
\end{equation}
over general symmetric matrices $\tbf{Q}$ and $\tbf{\hat{Q}}$. In the following we will assume simple ansatz for these matrices that allows to get analytic expressions in $n$ in order to take the derivative.

\subsection{RS entropy}
\label{appendix:replica_rs}

Let's compute the functional $S_n(\tbf{Q},\tbf{\hat{Q}})$ appearing in the free entropy eq.~\eqref{appendix:free_entropy2} in the simplest ansatz: the Replica Symmetric ansatz. This later assumes that all replica remain equivalent with a common overlap $q_0 = \frac{1}{N} \sum_{i=1}^N w_i^a w_i^b$ for $a\ne b$ and a norm $Q= \frac{1}{N} \sum_{i=1}^N w_i^a w_i^a$, leading to the following expressions of the matrices $\tbf{Q}$ and $\tbf{\hat{Q}} \in \mathbb{R}^{n\times n}$:
\begin{equation}
\begin{aligned}[c]
\tbf{Q} =\begin{pmatrix} 
 Q & q_0 & ... & q_0 \\
 q_0 & Q & ... & ...  \\
 ... &... & ... & q_0  \\
 q_0 &... & q_0 & Q    \\
\end{pmatrix}
\end{aligned}
\hspace{0.5cm}
\textrm{ and } 
\hspace{0.5cm}
\begin{aligned}[c]
\tbf{\hat{Q}}=\begin{pmatrix} 
 \hat{Q} & \hat{q}_0 & ... & \hat{q}_0\\
\hat{q}_0 &\hat{Q} & ... & ...  \\
 ... &... & ... & \hat{q}_0  \\
\hat{q}_0 &... & \hat{q}_0 & \hat{Q}\\  
\end{pmatrix} .
\end{aligned}
\end{equation}

Let's compute separately the terms involved in the functional $S_n(\tbf{Q},\tbf{\hat{Q}})$ eq.~\eqref{appendix:S_n}: the first is a trace term, the second a term of prior $\mathcal{I}^w_n$ and finally the third a term depending on the constraint $\mathcal{I}^z_n$.

\paragraph{Trace term} 
The trace term can be easily computed and takes the following form:
\begin{equation}
	\left.\frac{1}{2}\Tr(\tbf{Q\hat{Q}}) \right|_{\rm RS} =\frac{1}{2} \left( n Q\hat{Q} + n(n-1) q_0\hat{q}_0 \right).
\end{equation}

\paragraph{Prior integral} Evaluated at the RS fixed point, and using a gaussian identity also known as a Hubbard-Stratonovich transformation, the prior integral can be further simplified 
\begin{align}
	\left.\mathcal{I}^w_n (\tbf{\hat{Q}})\right|_{\rm RS} &= \int d\tbf{\tilde{w}} P_{\tilde{w}}(\tbf{\tilde{w}})  e^{ \frac{1}{2}\tbf{\tilde{w}}^{\intercal} \tbf{\hat{Q}} \tbf{\tilde{w}} } =   \int d\tbf{\tilde{w}} P_{\tilde{w}}(\tbf{\tilde{w}})  \exp{ \left( {\frac{(\hat{Q}- \hat{q}_0  )}{2}\sum_{a=1}^n (\tilde{w}^a)^2}\right)} \exp{ \left(\hat{q}_0 \left(  \sum_{a=1}^n  \tilde{w}^a  \right)^2 \right)}\\
	&= \int Dt \left [ \int d w P_w(w)  \exp{ \left( {\frac{(\hat{Q}- \hat{q}_0  )}{2}  w^2}+ t\sqrt{\hat{q}_0} w \right)}  \right]^n . 
\end{align}

\paragraph{Constraint integral} 
Recall the vector $\tbf{\tilde{z}}\sim P_{\tilde{z}} \triangleq \mathcal{N}(\tbf{0},\tbf{Q})$ follows a gaussian distribution with zero mean and covariance matrix $\tbf{Q}$. In the RS ansatz, the covariance can be rewritten as a linear combination of the identity $\tbf{I}$ and $\tbf{J}$ the matrix with all ones entries of size $n \times n $: $\left. \tbf{Q} \right|_{\rm RS} = (Q-q_0) \tbf{I} + q_0 \tbf{J}$, that allows to split the variable $z^a = \sqrt{q_0} t + \sqrt{Q-q_0} u^a  $ with  $t \sim \mathcal{N}(0,1)$ and $ \forall a,~ u_a \sim \mathcal{N}(0,1)$.  Finally, the constraint integral reads:
\begin{align}
\left. \mathcal{I}^z_n(\tbf{Q}) \right|_{\rm RS} &= \int  d\tbf{\tilde{z}} P_{\tilde{z}}(\tbf{\tilde{z}})  \mathcal{C} (\tbf{\tilde{z}}) = \int Dt    \int  \prod_{a=1}^n Du^a  \varphi\left(\sqrt{q_0} t + \sqrt{Q-q_0} u^a \right)\\
	 &=  \int Dt    \left[ \int  Du  \varphi\left(\sqrt{q_0} t + \sqrt{Q-q_0} u \right) \right]^n . 
\end{align}

\paragraph{Summary and RS free entropy $\phi_{\rm RS}$}
Finally putting pieces together, the functional $S_n$ taken at the RS fixed point has an explicit formula and dependency in $n$:
 
\begin{align}
	\left. S_n(\tbf{Q},\tbf{\hat{Q}}) \right|_{\rm RS} &= \left.\frac{1}{2}\Tr(\tbf{Q\hat{Q}}) -\log(\mathcal{I}_w^n (\tbf{\hat{Q}}))-\alpha\log \left(\mathcal{I}_z^n(\tbf{Q}) \right) \right|_{\rm RS}\\
	&\underset{n\to 0}{\simeq} \frac{1}{2} \left( n Q\hat{Q} + n(n-1) q_0\hat{q}_0 \right) - n \int Dt \log\left( \int dw P_w(w)  \exp{ \left( {\frac{(\hat{Q}- \hat{q}_0  )}{2}  w^2}+ t\sqrt{\hat{q}_0} w \right)}   \right) \\
	& -n\alpha  \int Dt  \log\left(   \int  Du  \varphi\left(y,\sqrt{q_0} t + \sqrt{Q-q_0} u  \right)  \right).
\end{align}
Finally taking the derivative with respect to $n$ and the $n\to 0$ limit, the RS free entropy has a simple expression
\begin{equation}
	\phi_{\rm RS}(\alpha) = \textrm{SP}_{q_0,\hat{q}_0} \left\{   -\frac{1}{2}Q\hat{Q} + \frac{1}{2}q_0\hat{q}_0+  \mathcal{I}^w_{\rm RS}(\hat{q}_0)   +\alpha  \mathcal{I}^z_{\rm RS}(q_0)    \right\},
\end{equation}
with $Q=\hat{Q}=1$ and the following notations,
\begin{equation}
	\begin{cases}
		\mathcal{I}^w_{\rm RS}(\hat{q}_0) \equiv  \int Dt\log\left(  \int dw P_w(w)  \exp{ \left( {\frac{(\hat{Q}- \hat{q}_0  )}{2}  w^2}+ t\sqrt{\hat{q}_0} w \right)}   \right) \vspace{0.5cm} \\
		\mathcal{I}^z_{\rm RS}(q_0) \equiv  \int Dt  \log\left(    \int  Dz  \varphi\left(\sqrt{q_0} t + \sqrt{Q-q_0} z  \right)  \right)
	\end{cases}.
\end{equation}

\subsection{1RSB entropy}
\label{appendix:replica_1RSB}
The free entropy eq.~\eqref{appendix:free_entropy} can also be evaluated at the simplest non trivial fixed point: the one step Replica Symmetry Breaking ansatz (1RSB). Instead assuming that replicas are equivalent, it assumes that the symmetry between replica is broken and that replicas are clustered in different states, with inner overlap $q_1$ and outer overlap $q_0$. Translating this in a matrix formulation, the matrices can be expressed as
\begin{equation}
		\tbf{Q} = q_0 \tbf{J}_n + \left( q_1 - q_0 \right) \tbf{I}_{\frac{n}{m}} \otimes \tbf{J}_{m} +  \left( Q - q_1 \right) \tbf{I}_n 
		\hspace{0.2cm}
		\textrm{  and  }\hspace{0.2cm}
		\tbf{\hat{Q}} = \hat{q}_0 \tbf{J}_n + \left( \hat{q}_1 - \hat{q}_0 \right) \tbf{I}_{\frac{n}{m}} \otimes \tbf{J}_{m} +  \left( \hat{Q} - \hat{q}_1 \right) \tbf{I}_n \,.
\end{equation}

\paragraph{Trace term} 
Again, the trace term can be easily computed
\begin{equation}
	\left.\frac{1}{2}\Tr(\tbf{Q\hat{Q}}) \right|_{\rm 1RSB} =\frac{1}{2} \left( n Q\hat{Q} + n(m-1)q_1\hat{q}_1 + n(n-m)q_0\hat{q}_0 \right).
	\label{appendix:1RSB_Tr}
\end{equation}

\paragraph{Prior integral}
Separating replicas with different overlaps, the prior integral can be written as
\begin{align}
	\left.\mathcal{I}^w_n (\tbf{\hat{Q}})\right|_{\rm 1RSB} &=  \int d\tbf{\tilde{w}} P_{\tilde{w}}(\tbf{\tilde{w}})  e^{  \frac{(\hat{Q}- \hat{q}_1  )}{2}\sum_{a=1}^n (\tilde{w}^a)^2 + \frac{(\hat{q}_1-\hat{q}_0)}{2} \sum_{k=1}^{\frac{n}{m}} \sum_{a,b=(k-1)m + 1 }^{km} \tilde{w}^a \tilde{w}^b + \frac{\hat{q}_0}{2} \left(\sum_{a=1}^n \tilde{w}^a \right)^2  } \\
	&= \int Dt_0 \left[\int Dt_1\left[ \int dw P_w(w) \exp\left( \frac{(\hat{Q}- \hat{q}_1  )}{2} w^2 + \left(\sqrt{\hat{q}_0}t_0+\sqrt{\hat{q}_1-\hat{q}_0}t_1 \right)w \right) \right ]^{m} \right]^{\frac{n}{m}}
	\label{appendix:1RSB_Iw}
\end{align}

\paragraph{Constraint integral} Again the vector $\tbf{\tilde{z}} \sim P_{\tilde{z}} \triangleq \mathcal{N}(\tbf{0},\tbf{Q})$  follows a gaussian vector with zero mean and covariance $ \left. \tbf{Q} \right|_{\rm 1RSB} = q_0 \tbf{J}_n + \left( q_1 - q_0 \right) \tbf{I}_{\frac{n}{m}} \otimes \tbf{J}_{m} +  \left( Q - q_1 \right) \tbf{I}_n $. The gaussian vector of covariance $\left. \tbf{Q} \right|_{\rm 1RSB}$ can be decomposed in a sum of normal gaussian vectors $t_0 \sim \mathcal{N}(0,1)$, $ \forall k \in [1: \frac{n}{m}], t_k \sim \mathcal{N}(0,1)  $ and $ \forall a \in [(k-1) m +1 : km ]$, $u_a \sim \mathcal{N}(0,1)$: $ z^a = \sqrt{q_0} t_0 + \sqrt{q_1-q_0} t_k  + \sqrt{Q-q_1} u_{a}  $. Finally the constraint integral reads
\begin{align}
\left.  \mathcal{I}^z_n(\tbf{Q}) \right|_{\rm 1RSB}&= \int Dt_0  \int \prod_{k=1}^{\frac{n}{m}}  Dt_k \int \prod_{a=(k-1)m+1}^{km}  Du_a   \varphi(\sqrt{q_0} t_0 + \sqrt{q_1-q_0} t_k  + \sqrt{Q-q_1} u_{a})\\
&= \int Dt_0  \left[ \int Dt_1  \left[ \int Du   \varphi(\sqrt{q_0} t_0 + \sqrt{q_1-q_0} t_1  + \sqrt{Q-q_1} u)\right]^m \right]^{\frac{n}{m}}.
\label{appendix:1RSB_Iz}
\end{align}
 
\paragraph{Summary and 1RSB free entropy $\phi_{\rm 1RSB}$}
Gathering the previous computations eq.~(\ref{appendix:1RSB_Tr}, \ref{appendix:1RSB_Iw}, \ref{appendix:1RSB_Iz}), the functional $S_n$ evaluated at the 1RSB fixed point reads:
\begin{align}
	\left. S_n(\tbf{Q},\tbf{\hat{Q}}) \right|_{\rm 1RSB} &= \left.\frac{1}{2}\Tr(\tbf{Q\hat{Q}}) -\log(\mathcal{I}_w^n (\tbf{\hat{Q}}))-\alpha\log \left(\mathcal{I}_z^n(\tbf{Q}) \right) \right|_{\rm 1RSB}\\
	&\underset{n\to 0}{\simeq}
	\frac{1}{2} \left( n Q\hat{Q} + n(m-1)q_1\hat{q}_1 + n(n-m)q_0\hat{q}_0 \right) \\
	&-\frac{n}{m} \int Dt_0 \log\left(\int Dt_1\left[ \int d\tilde{w} P_w(\tilde{w}) \exp\left( \frac{(\hat{Q}- \hat{q}_1  )}{2} \tilde{w}^2 + \left(\sqrt{\hat{q}_0}t_0+\sqrt{\hat{q}_1-\hat{q}_0}t_1 \right)\tilde{w} \right) \right ]^{m} \right) \\
	& -\alpha\frac{n}{m} \int dy \int Dt_0  \log\left(  \int Dt_1  \left[ \int Du   \varphi(y,\sqrt{q_0} t_0 + \sqrt{q_1-q_0} t_1  + \sqrt{Q-q_1} u)\right]^m  \right).
\end{align}

Let's introduce the replicated free entropy following \cite{25}. We consider $m$ reals replicas of the same system and we imagine we put a small field, that allows the $m$ replicas to fall in the same state. The replicated free entropy is the free entropy corresponding to these $m$ uncorrelated copies in the limit of zero coupling. To compute it, we consider $n'=\frac{n}{m}$ replicas. Denoting $\tbf{q}=(q_0,q_1)$ and $\tbf{\hat{q}}=(\hat{q}_0,\hat{q}_1)$, the replicated free entropy reads as $m$ times the free entropy of $n$ replicas with 1RSB structure:
\begin{align}
	\Phi^{\rm 1RSB}(\alpha) :&= \left( \lim_{N \to \infty} \frac{1}{ N} \mathbb{E}_{\tbf{X}}\left[\log( \mathcal{Z}_m(\tbf{X} )\right] \right)  \simeq  \lim_{N \to \infty} \frac{1}{N} \lim_{n' \rightarrow 0} \frac{\partial \log \left( \mathbb{E}_{\tbf{X}} [ \mathcal{Z}^{mn'}(\tbf{X}) ] \right)}{\partial n'} \\
	& = m \left( \lim_{N \to \infty} \lim_{n\to 0} \frac{1}{ N} \frac{ \partial \log\left(\mathbb{E}[ \mathcal{Z}^{n}(\tbf{X}) ]_{\tbf{X}} \right) }{\partial n}\right) = m \left( - \text{SP}_{ \tbf{Q}, \tbf{\hat{Q}} } \left\{\lim_{n\rightarrow 0} \frac{\partial S_n(\bold{Q},\bold{\hat{Q}})}{\partial  n} \right\} \right)\\
	&= \underset{\tbf{q} , \tbf{\hat{q}}}{\textrm{SP}} \left\{  \frac{m}{2} \left(  q_1\hat{q}_1 - Q\hat{Q} \right) + \frac{m^2}{2} \left(q_0\hat{q}_0 - q_1\hat{q}_1 \right)   + m \mathcal{I}^w_{\rm 1RSB}(\tbf{\hat{q}})   +\alpha m \mathcal{I}^{z}_{\rm 1RSB}(\tbf{q})    \right\}\,.
\end{align}
with $\tbf{t}=(t_0,t_1)$, $g_0^w$ and $f_0^z$ defined in eq.~\eqref{main:f_i_g_i} and
\begin{equation}
		\mathcal{I}^w_{\rm 1RSB}(\tbf{\hat{q}}) = \frac{1}{m}
                \int Dt_0 \log\left(\int Dt_1 g_0^w\left(\tbf{t},\tbf{\hat{q}} \right )^{m}\right) \hspace{0.2cm}
		 \textrm{ and }\hspace{0.2cm}
		\mathcal{I}^{z}_{\rm 1RSB}(\tbf{q}) = \frac{1}{m}
                \int Dt_0  \log\left(  \int Dt_1  f_0^z\left(\tbf{t},\tbf{q}\right)^m \right)\,.
\end{equation}

\subsection{RS Stability}
\label{appendix:AT_stability}
\subsubsection{De Almeida Thouless RS Stability}
The stability of a given saddle point ansatz is related to the positivity the hessian of the functional $S_n$. This stability analysis has first been done by de Almeida Thouless and following \cite{22,1,engel2001statistical}, replicons eigenvalues of the RS ansatz $\lambda_3^A$ and $\lambda_3^B$ can be expressed as functions of $\{g_i^w,f_i^z\}_{i=0}^2$ defined in eq.~\eqref{main:f_z_g_w_rs}:
\begin{align}
		\lambda_3^A(q_0)= \frac{1}{(Q-q_0)^2} \int Dt \frac{\left(f_0^{z}(f_0^{z}-f_2^{z}) + (f_1^{z})^2 \right)^2}{(f_0^{z})^4}(t,q_0)\, , \hspace{0.2cm} \textrm{ and } \hspace{0.2cm}
		\lambda_3^B (\hat{q}_0) = \int Dt \frac{ \left(g_0^{w}g_2^{w} -(g_1^{w})^2  \right)^2 }{(g_0^{w})^4}(t,\hat{q}_0) 
		\, .		
\end{align}
The instability AT-line is defined when the determinant of the hessian vanishes that translates as an implicit equation over $\alpha$, where $q_0,\hat{q}_0$ are solution of the saddle point equations eq.~\eqref{main:phi_RS} at $\alpha=\alpha_{AT}$: 
\begin{align}
	\frac{1}{\alpha_{AT}} &=   \lambda_{3}^A \left(q_0(\alpha_{AT}),\beta \right) \lambda_{3}^B \left(\hat{q}_0(\alpha_{AT}) \right) \,.
\end{align}  

However for $\alpha < \alpha_{AT}$, $(q_0,\hat{q}_0)=(0,0)$ is the only solution. Using $\{\tilde{f}^z_i,\tilde{g}^w_i \}_{i=0}^2$ defined eq.~\eqref{appendix:fgtilde}, this expression simplifies because of the symmetry of the prior distribution $P_w$ and the constraints $\varphi$ in the rectangle and $u-$function cases. In fact the symmetry imposes $\tilde{f}_1^{z}=0$ and $\tilde{g}_1^{w} = 0$ and the condition reads:
\begin{align}
	\frac{1}{\alpha_{AT}} &= \left(\frac{\tilde{f}_2^{z}-\tilde{f}_0^{z}}{\tilde{f}_0^{z}}\right)^2 \left(\frac{\tilde{g}_2^{w} }{\tilde{g}_0^{w}}\right)^2 \, .
	\label{appendix:AT_line}
\end{align}  

\subsubsection{Existence and stability of the RS fixed point $(q_0,\hat{q}_0)=(0,0)$} 
We provide an alternative approach to get the instability condition of the RS solution for symmetric prior and constraint. In this symmetric case, the stability can be derived from the existence and stability of the symmetric fixed point $(q_0,\hat{q}_0)=(0,0)$. Let's define
\begin{align}
	\begin{cases}
	F(q_0) \equiv \alpha  \int Dt \frac{(f_1^{z})^2-2t\sqrt{q_0}f_0^{z}f_1^{z} +q_0 t^2 (f_0^{z})^2}{(1-q_0)^2  (f_0^{z})^2}(t,q_0) \,, \vspace{0.3cm} \\
	G( \hat{q}_0) \equiv  \int Dt \frac{g_2^{w} -  t\hat{q}_0^{-1/2}g_1^{w}}{g_0^{w}}(t,\hat{q}_0) \,,
	\end{cases}
	\textrm{ with }  
\begin{cases}
\tilde{f}_i^{z}(y) \equiv \int Dz z^i \varphi(z) \,, \vspace{0.3cm} \\
\tilde{g}_i^{w} \equiv \int dw w^i P_w(w) e^{ \frac{w^2}{2}}\,.
\end{cases}
\label{appendix:fgtilde}
\end{align}

In fact the saddle point equations at the RS fixed point eq.~\eqref{main:phi_RS} can be written using the functions $F, G$, and can be reduced to a single fixed point equation over $q_0$:
\begin{align}
	\begin{cases}
         q_0=G(\hat{q}_0)\,,  \vspace{0.3cm} \\
         \hat{q}_0=F(q_0)\,,
     \end{cases}
     \Rightarrow
     \begin{cases}
         q_0=G \circ F (q_0) \equiv H(q_0)\, . \\
\end{cases}
\label{appendix:H_q0}
\end{align}

As stressed above, the RS stability is equivalent to the existence and stability of the fixed point $q_0=0$. According to that, let's compute the stability of the above fixed point equation eq.~\eqref{appendix:H_q0}. Computing $F, F', G, G'$ in the limit $(q_0,\hat{q}_0) \to (0,0)$, expanding $\{f_i^z$,$g_i^w\}_i$ as functions of $\{\tilde{f}^z_i,\tilde{g}^w_i \}_i$ and finally using the symmetry that implies $\tilde{f}_1^{z}=0$ and $\tilde{g}_1^{w} = 0$:
\begin{equation}
\begin{cases}
		F(q_0) \underset{q_0 \to 0}{=} \alpha    \left [ \left(\frac{\tilde{f}_1^{z}}{\tilde{f}_0^{z}}\right)^2 +q_0\left( \frac{(\tilde{f}_2^{z}-\tilde{f}_0^{z})^2}{(\tilde{f}_0^{z})^2} +3\frac{(\tilde{f}_1^{z})^4}{(\tilde{f}_0^{z})^4} -4 \frac{(\tilde{f}_1^{z})^2(\tilde{f}_2^{z}-\tilde{f}_0^{z})}{(\tilde{f}_0^{z})^3}   \right)  + \mathcal{O}(q_0^2) \right] {\sim} \alpha q_0    \left(\frac{\tilde{f}_2^{z}-\tilde{f}_0^{z}}{\tilde{f}_0^{z}}\right)^2 \underset{q_0 \to 0}{\longrightarrow} 0 \,, \vspace{0.2cm} \\
		\frac{\partial F}{\partial q_0} (q_0) \underset{q_0 \to 0}{=} \alpha  \left[ \left(\frac{\tilde{f}_2^{z}-\tilde{f}_0^{z}}{\tilde{f}_0^{z}}\right)^2 + \left(\frac{\tilde{f}_1^{z}}{\tilde{f}_0^{z}}\right)^2 \left (3\frac{(\tilde{f}_1^{z})^2}{(\tilde{f}_0^{z})^2} -4 \frac{(\tilde{f}_2^{z}-\tilde{f}_0^{z})}{\tilde{f}_0^{z}}\right)   + \mathcal{O}(q_0) \right] \underset{q_0 \to 0}{\longrightarrow}  \alpha    \left(\frac{\tilde{f}_2^{z}-\tilde{f}_0^{z}}{\tilde{f}_0^{z}}\right)^2 \,, \vspace{0.2cm} \\
		G( \hat{q}_0) \underset{\hat{q}_0 \to 0}{=} \left(\frac{\tilde{g}_1^{w} }{\tilde{g}_0^{w}}\right)^2  + \hat{q}_0\left( \left(\frac{\tilde{g}_2^{w} }{\tilde{g}_0^{w}}\right)^2 + \frac{\tilde{g}_1^{w}}{\tilde{g}_0^{w}} \left(3 \left(\frac{\tilde{g}_1^{w} }{\tilde{g}_0^{w}}\right)^3-4\frac{\tilde{g}_1^{w}\tilde{g}_2^{w}}{(\tilde{g}_0^{w})^2} \right) \right) + \mathcal{O}( \hat{q}_0^{3/2}) \underset{ \hat{q}_0 \to 0}{\longrightarrow} 0  \,, \vspace{0.2cm} \\
		\frac{\partial G}{\partial \hat{q}_0}(\hat{q}_0)  \underset{\hat{q}_0 \to 0}{=}  \left(\frac{\tilde{g}_2^{w} }{\tilde{g}_0^{w}}\right)^2 + \frac{\tilde{g}_1^{w}}{\tilde{g}_0^{w}} \left(3 \left(\frac{\tilde{g}_1^{w} }{\tilde{g}_0^{w}}\right)^3-4\frac{\tilde{g}_1^{w}\tilde{g}_2^{w}}{(\tilde{g}_0^{w})^2} \right) + \mathcal{O}(\sqrt{ \hat{q}_0}) \underset{ \hat{q}_0 \to 0}{\longrightarrow} \left(\frac{\tilde{g}_2^{w} }{\tilde{g}_0^{w}}\right)^2 \,.
\end{cases}
\end{equation}
Finally, the existence and stability conditions of the fixed point $(q_0,\hat{q}_0)=(0,0)$  
translate as an explicit condition over $\alpha$ that defines $\alpha_{AT}$
\begin{equation}
	\begin{cases}
		H(q_0) = G \circ F(q_0) \underset{q_0 \to 0}{\to} 0 \vspace{0.3cm} \\
		\left. \frac{\partial H}{\partial q_0}\right|_{q_0=0} =  \left. \frac{\partial G}{\partial \hat{q}_0}\right|_{\hat{q}_0=0}  \left.\frac{\partial F}{\partial q_0}\right|_{q_0=0} \leq 1 \,, 
	\end{cases}
	\Rightarrow 
	\hspace{0.5cm}
	\alpha \leq  \left[   \left(\frac{\tilde{f}_2^{z}-\tilde{f}_0^{z}}{\tilde{f}_0^{z}}\right)^2 \left(\frac{\tilde{g}_2^{w} }{\tilde{g}_0^{w}}\right)^2 \right]^{-1} \equiv \alpha_{AT}\,.
	\label{stability}
\end{equation}

\subsection{Moments at finite temperature}
\label{moments_finiteT}
In this section we generalize the definition of the partition function for any temperature $T$. The energy of a configuration $\tbf{w}$ is defined as the number of unsatisfied constraints and the corresponding partition function is defined by $\mathcal{Z}(\tbf{X},T) =\sum_{\tbf{w}\in \{\pm1\}^N} e^{-\mathcal{E}(\tbf{w})/T} $. In particular for the rectangle and $u-$function constraints, the partition functions at temperature $T$ read
\begin{align}
	\mathcal Z_r(\tbf{X},T) = \displaystyle \sum_{ \textbf{w} \in
  \{\pm 1\}^N} \prod_{\mu = 1}^M   e^{-\frac{1}{T} \left(1 -
  \mathbbm{1}_{\left |\displaystyle z_\mu (\textbf{w}) \right| \le K }
  \right)} \hspace{0.2cm} \textrm{ and } \hspace{0.2cm} \mathcal
  Z_u(\tbf{X},T) = \displaystyle \sum_{ \textbf{w} \in \{\pm 1\}^N}
  \prod_{\mu = 1}^M   e^{-\frac{1}{T} \left(1 -  \mathbbm{1}_{\left
  |\displaystyle z_\mu (\textbf{w}) \right| \ge K } \right)}\, .
\end{align}
We define the probabilities that constraints are satisfied at temperature $T$:
\begin{equation}
\begin{cases}
		p_{r,K,T} \equiv \int Dz e^{-\frac{1}{T} \left(1 -
                    \mathbbm{1}_{\left |\displaystyle z \right| \le K
                    } \right)} =  e^{-\frac{1}{T}} +
                (1-e^{-\frac{1}{T}}) p_{r,K} \, ,\vspace{0.2cm} \\
	p_{u,K,T} \equiv \int Dz e^{-\frac{1}{T} \left(1 -
            \mathbbm{1}_{\left |\displaystyle z \right| \ge K }
          \right)} = e^{-\frac{1}{T}} + (1-e^{-\frac{1}{T}}) p_{u,K}
        \, , \vspace{0.2cm} \\
	p_{s,K,T} \equiv \int Dz e^{-\frac{1}{T} \left(1 -
            \mathbbm{1}_{\displaystyle z \ge K } \right)} =
        e^{-\frac{1}{T}} + (1-e^{-\frac{1}{T}}) p_{s,K} \, .\vspace{0.1cm} \\
\end{cases}	
\end{equation}

\subsubsection{First moment at finite temperature}
 Let  $\cE^r(N,M,T)$ the event that $\mathcal Z_r(\tbf{X},T)\ge1$.
Let's compute the first moment in the rectangle case,
\begin{align}
	\bbP[ \cE^r (N, \alpha N,T ) ] &\le \E [\mathcal{Z}_{r}(\tbf{X}(N,\alpha N),T)] = 2^N \E \left [ \prod_{\mu =1}^{\alpha N} e^{-\frac{1}{T} \left(1 -  \mathbbm{1}_{\left |\displaystyle z_\mu (\textbf{1}) \right| \le K } \right)}  \right] \\
&= 2^N p_{r,K,T}^{\alpha N} = \exp(N (\log(2)+\alpha \log( p_{r,K,T})  )) \,.
\end{align}
and this derivation holds similarly for the step and $u-$function.

\subsubsection{Second moment at finite temperature}
\label{}
Again we show the computation for the rectangle and it can be done similarly for the $u-$function. 

\paragraph{Expression of $F_{r,K,\alpha,T}$}
\begin{align}
\E [\mathcal Z_r (\tbf{X}(N,\alpha N),T)^2   ] &=  \sum_{\tbf{w_1},\tbf{w_2}\in \{ \pm 1\}^N} \E \left [\prod_{\mu=1}^{\alpha N} e^{-\frac{1}{T} \left(1 -  \mathbbm{1}_{\left |\displaystyle z_\mu (\tbf{w_1}) \right| \le K } \right)} e^{-\frac{1}{T} \left(1 -  \mathbbm{1}_{\left |\displaystyle z_\mu (\tbf{w_2}) \right| \le K } \right)} \right] \\
&= 2^N \sum_{\tbf{w}\in \{ \pm 1\}^N} \prod_{\mu=1}^{\alpha N } \E \left [ e^{-\frac{1}{T} \left \{ \left(1 -  \mathbbm{1}_{\left |\displaystyle z_\mu (\tbf{1}) \right| \le K } \right) + \left(1 -  \mathbbm{1}_{\left |\displaystyle z_\mu (\tbf{w}) \right| \le K } \right) \right \} } \right] \\
&= 2^N \sum_{l=0}^N \binom{N}{l} q_{r,K,T}(l/N)^{\alpha N}  \equiv  \exp( N (\log(2) + F_{r,K,\alpha,T} ) ) \, ,
\end{align}
where we defined $q_{r,K,T}$ the probability that two standard Gaussians with correlation $\beta$ are both at most $K$ in absolute value at temperature $T$. Defining $\rho(\beta) = 1- 2 \beta$ and 
\begin{align}
	\mathcal{I}_{\alpha_1,\beta_1}^{\alpha_2,\beta_2} (\rho) \equiv \int_{\alpha_1}^{\beta_1} \int_{\alpha_2}^{\beta_2}  dx dy \frac{e^{-\frac{1}{2} (x^2 + y^2 + 2\rho xy) }}{2\pi \sqrt{1-\rho^2}} = \frac{1}{2\pi} \int_{\alpha_2}^{\beta_2} \int_{\frac{\alpha_1 + \rho y }{\sqrt{1-\rho^2}} }^{\frac{\beta_1 + \rho y }{\sqrt{1-\rho^2}} }  dy dx e^{-\frac{y^2+x^2}{2}}\,,
\end{align}
the function $ F_{r,K,\alpha,T}$ at finite temperature can be written
\begin{align*}
	 F_{r,K,\alpha,T} &= H(\beta)  +  \alpha \log q_{r,K,T}(\beta)\,,
\end{align*}
where
\begin{align}
	&q_{r,K,T}(\beta)\equiv \int_{\mathbbm{R}^2}  dx dy \frac{e^{-\frac{1}{2} (x^2 + y^2 + 2\rho(\beta) xy) }}{2\pi \sqrt{1-\rho(\beta)^2}} e^{-\frac{1}{T} \left( \left(1 -  \mathbbm{1}_{\left |\displaystyle z_\mu (\tbf{1}) \right| \le K } \right) + \left(1 -  \mathbbm{1}_{\left |\displaystyle z_\mu (\tbf{w}) \right| \le K } \right) \right)} \\
	&= \mathcal{I}_{-K,K}^{-K,K}  + e^{-\frac{1}{T}}  \left( \mathcal{I}_{-\infty,-K}^{-K,K} + \mathcal{I}_{K,+\infty}^{-K,K}  + \mathcal{I}_{-K,K}^{-\infty,-K}  + \mathcal{I}_{-K,K}^{K,+\infty} \right) + e^{-\frac{2}{T}}  \left( \mathcal{I}_{-\infty,-K}^{-\infty,-K} + \mathcal{I}_{-\infty,-K}^{K,+\infty}  + \mathcal{I}_{K,+\infty}^{-\infty,-K}  + \mathcal{I}_{K,+\infty}^{K,+\infty}  \right)\,.
	\label{appendix:q_r_T}
\end{align}

\paragraph{Expression of $\partial_{\beta}F_{r,K,\alpha,T}$ \\}

To compute the derivative of $q_{r,K,T}$, we first introduce
\begin{align*}
	\mathcal{G}_{\gamma}^{\alpha_2,\beta_2}(\rho) \equiv \frac{1}{2\pi} \int_{\alpha_2}^{\beta_2} dy e^{-\frac{y^2}{2}} e^{-\frac{1}{2}\frac{(\gamma + \rho y )}{1-\rho^2} } (y + \gamma \rho) \,.
\end{align*}
The derivative of each integral involved in eq.~\eqref{appendix:q_r_T} can be easily computed as
\begin{align}
	\partial_{\beta} \mathcal{I}_{\alpha_1,\beta_1}^{\alpha_2,\beta_2} (\rho(\beta)) = -\frac{1}{4(\beta (1-\beta))^{3/2}} \left( \mathcal{G}_{\beta_1}^{\alpha_2, \beta_2} -\mathcal{G}_{\alpha_1}^{\alpha_2, \beta_2} 	\right)(\rho(\beta))\,.
\end{align}
Hence taking the derivative of each term of the form $\mathcal{I}_{\alpha_1,\beta_1}^{\alpha_2,\beta_2}$ and simplifying it, the probability $q_{r,K,T}$ reads:
\begin{align*}
	q_{r,K,T}(\beta) &= -\frac{1}{4(\beta (1-\beta))^{3/2}} \left( \mathcal{G}_{K}^{-K,K} -\mathcal{G}_{-K}^{-K,K} 	\right)(\rho) (1-e^{-1/T})^2 
	= \frac{(1-e^{-1/T})^2}{\pi \sqrt{\beta(1-\beta)}} \left( e^{-\frac{K^2}{2(1-\beta)}} \left( e^{\frac{(2\beta-1)K^2}{2(1-\beta)\beta}} -1  \right)  \right) \,.
\end{align*}
In the end, the derivative of the second moment can be evaluated for $\beta = 0$ and $\beta=1$ at all temperature $T$ :
\begin{align}
	\frac{\partial F_{r,K,\alpha,T}}{\partial \beta}(\beta) &= \log\left (\frac{1-\beta}{\beta} \right) +  \frac{\alpha}{q_{r,K,T}}\frac{\partial q_{r,K,T}(\beta)}{\partial \beta}  \\
	&=  \log\left (\frac{1-\beta}{\beta} \right) +  \frac{\alpha}{ q_{r,K,T}(\beta)} \frac{(1-e^{-1/T})^2}{\pi \sqrt{\beta(1-\beta)}} \left( e^{-\frac{K^2}{2(1-\beta)}} \left( e^{\frac{(2\beta-1)K^2}{2(1-\beta)\beta}} -1  \right)  \right) \xrightarrow[\beta \to 1/2 \pm 1/2]{} \pm \infty \,.
	\label{appendix:derivative_Fr}
\end{align}
In particular at $T=0$,
\begin{align}
	\frac{\partial F_{r,K,\alpha}}{\partial \beta}(\beta) &=
                                                                \log\left
                                                                (\frac{1-\beta}{\beta}
                                                                \right)
                                                                +
                                                                \frac{\alpha}{
                                                                q_{r,K,T}(\beta)}
                                                                \frac{1}{\pi
                                                                \sqrt{\beta(1-\beta)}}
                                                                \left(
                                                                e^{-\frac{K^2}{2(1-\beta)}}
                                                                \left(
                                                                e^{\frac{(2\beta-1)K^2}{2(1-\beta)\beta}}
                                                                -1
                                                                \right)
                                                                \right)
                                                                \, . 
\end{align}

\paragraph{Expression of $\partial_{\beta}F_{u,K,\alpha,T}$\\}
Adapting the previous steps and using
\begin{align*}
		&q_{u,K,T}(\beta) \equiv \int_{\mathbbm{R}^2}  dx dy \frac{e^{-\frac{1}{2} (x^2 + y^2 + 2\rho(\beta) xy) }}{2\pi \sqrt{1-\rho(\beta)^2}} e^{-\frac{1}{T} \left( \left(1 -  \mathbbm{1}_{\left |\displaystyle z_\mu (\tbf{1}) \right| \le K } \right) + \left(1 -  \mathbbm{1}_{\left |\displaystyle z_\mu (\tbf{w}) \right| \le K } \right) \right)} \\
&= \left( \mathcal{I}_{-\infty,-K}^{-\infty,-K} + \mathcal{I}_{-\infty,-K}^{K,+\infty}  + \mathcal{I}_{K,+\infty}^{-\infty,-K}  + \mathcal{I}_{K,+\infty}^{K,+\infty}  \right)   + e^{-\frac{1}{T}}  \left( \mathcal{I}_{-\infty,-K}^{-K,K} + \mathcal{I}_{K,+\infty}^{-K,K}  + \mathcal{I}_{-K,K}^{-\infty,-K}  + \mathcal{I}_{-K,K}^{K,+\infty} \right) + e^{-\frac{2}{T}}  \left( \mathcal{I}_{-K,K}^{-K,K} \right)\\
		&= q_{r,K,-T}e^{-\frac{2}{T}}\,,
\end{align*}
and eq.~\eqref{appendix:derivative_Fr} the derivative for the $u-$function is straightforward to compute and is given by
\begin{align*}
	\frac{\partial F_{u,K,\alpha,T}}{\partial \beta}(\beta) &= \log\left (\frac{1-\beta}{\beta} \right) +  \frac{\alpha}{q_{u,K,T}(\beta)}\frac{\partial q_{u,K,T}}{\partial \beta}(\beta)  \\
	&=  \log\left (\frac{1-\beta}{\beta} \right)  + \frac{\alpha}{ q_{u,K,T}(\beta)} \frac{(e^{-1/T}-1)^2}{\pi \sqrt{\beta(1-\beta)}} \left( e^{-\frac{K^2}{2(1-\beta)}} \left( e^{\frac{(2\beta-1)K^2}{2(1-\beta)\beta}} -1  \right)  \right)   \\
	& \xrightarrow[\beta \to 1/2 \pm 1/2]{} \pm \infty \,.
\end{align*}

\end{document}